\PassOptionsToPackage{dvipsnames}{xcolor}

\documentclass[12pt,a4paper]{article}

\usepackage{graphicx,xcolor,textpos}

\graphicspath{{./img/}{./Images/}}

\usepackage{}
\usepackage{helvet}
\usepackage[english]{babel}
\usepackage[normalem]{ulem}
\usepackage{amsmath}
\usepackage{amsthm}
\usepackage{thmtools,thm-restate}
\usepackage{bbm}
\usepackage{amssymb}
\usepackage{hyperref}
\usepackage{relsize}
\usepackage[left=1in, right=1in, top=1in, bottom=1in]{geometry}
\usepackage{physics}
\usepackage[shortlabels]{enumitem}
\usepackage{mathtools}
\usepackage{changepage}
\usepackage{caption}
\usepackage{subcaption}
\usepackage{verbatim}
\usepackage{url}
\usepackage{standalone}
\usepackage{tikz}
\usetikzlibrary{calc,patterns,angles,quotes}
\usepackage{dsfont}
\usepackage{mathrsfs}
\usepackage[affil-it]{authblk}
\usepackage{titlesec}
\usepackage[countmax]{subfloat}

\makeatletter
\newcommand\Label[1]{&\refstepcounter{equation}(\theequation)\ltx@label{#1}&}
\makeatother

\newcommand{\stkout}[1]{\ifmmode\text{\sout{\ensuremath{#1}}}\else\sout{#1}\fi}

\renewcommand{\d}{\text{d}}

\DeclareMathOperator{\reg}{reg}
\DeclareMathOperator{\C}{\mathbb{C}}
\DeclareMathOperator{\End}{End}

\DeclareMathOperator{\range}{range}

\newcommand{\id}{\mathrm{id}}

\let\originalleft\left
\let\originalright\right
\renewcommand{\left}{\mathopen{}\mathclose\bgroup\originalleft}
\renewcommand{\right}{\aftergroup\egroup\originalright}
\let\al\alpha

\newcommand{\KK}{\mathcal K}

\newcommand{\HH}{\mathcal H}
\newcommand{\ZZ}{\mathbb Z}

\renewcommand{\AA}{\mathcal A}

\newcommand{\NN}{\mathbb{N}}

\DeclareMathOperator{\Fusion}{\Phi}

\DeclareMathOperator{\supp}{\mathrm{supp}}

\DeclareMathOperator{\ano}{\Omega}

\DeclareMathOperator{\proj}{\mathrm{proj}}

\DeclareMathOperator{\rev}{\mathrm{rev}}

\newcommand{\sfA}{\mathsf{A}}
\newcommand{\sfB}{\mathsf{B}}
\newcommand{\sfC}{\mathsf{C}}
\providecommand{\edit}[1]{#1}

\newcommand{\sfU}{\edit{\mathsf{U}}}
\newcommand{\sfV}{\edit{\mathsf{V}}}
\newcommand{\sfD}{\edit{\mathsf{D}}}
\newcommand{\sfTheta}{\mathsf{\Theta}}

\newcommand{\ie}{i.e.\;}

\newcommand{\caA}{\mathcal{A}}
\newcommand{\caB}{\mathcal{B}}
\newcommand{\caC}{\mathcal{C}}
\newcommand{\caD}{\mathcal{D}}

\newcommand{\caP}{\mathcal{P}}
\newcommand{\caU}{\mathcal{U}}
\newcommand{\bbZ}{\mathbb{Z}}

\newcommand{\caI}{\mathcal{I}}
\newcommand{\ind}{\mathrm{ind}}

\newcommand{\diam}{\mathrm{diam}}

\renewcommand{\AA}{\mathcal A}

\newcommand{\spacing}{\vspace{.5cm}}

\DeclareMathOperator{\I}{\mathbf{1}}
\DeclareMathOperator{\Z}{\mathbb{Z}}

\newcommand{\Ad}[1]{\textrm{Ad}\left(#1\right)}
\newcommand{\Aut}[1]{\textrm{Aut}\left(#1\right)}

\newtheoremstyle{exampstyle}
{15pt} 
{15pt} 
{\itshape} 
{} 
{\bfseries} 
{.} 
{.5em} 
{} 

\theoremstyle{exampstyle}

\newtheorem{theorem}{Theorem}[section]

\newtheorem{lemma}[theorem]{Lemma}
\newtheorem{remark}[theorem]{Remark}

\newtheorem{proposition}[theorem]{Proposition}

\newenvironment{proofof}[1][]{\emph{Proof of #1 : }}{\hfill$\qedsymbol$\vskip.5cm}

\DeclareMathOperator{\dist}{\mathrm{dist}}

\DeclareMathOperator{\QCA}{\mathsf{QCA}}
\DeclareMathOperator{\Sym}{\mathsf{Sym}}
\DeclareMathOperator{\sEnt}{\mathsf{sEnt}}
\DeclareMathOperator{\SE}{\mathsf{SE}}
\DeclareMathOperator{\sSPT}{\mathsf{SPT}^{\mathsf{SE}}}
\DeclareMathOperator{\sSta}{\mathcal{S}}

\newcommand{\J}{y}
\newcommand{\Left}[1]{L_{#1}}

\usepackage{tikz-cd}
\usetikzlibrary{positioning,fit,calc}

\usepackage{graphicx} 
\usepackage{color}
\usepackage{tikz}
\usepackage{tikz-cd}
\usetikzlibrary{decorations.pathreplacing, calligraphy}

\usepackage{tcolorbox}

\newcommand{\iso}[2]{#1^{(#2)}}

\title{A complete classification of 2d symmetry protected states with symmetric entanglers}

\author[1]{Alex Bols \thanks{email: \href{abols01@phys.ethz.ch}{abols01@phys.ethz.ch}}}
\author[2]{Wojciech De Roeck \thanks{email: \href{wojciech.deroeck@kuleuven.be} {wojciech.deroeck@kuleuven.be}}}
\author[3]{Michiel De Wilde \thanks{email: \href{michiel.dewilde@ist.ac.at}{michiel.dewilde@ist.ac.at}}}
\author[2]{Bruno de O. Carvalho \thanks{email: \href{bruno.oliveira@kuleuven.be}{bruno.oliveira@kuleuven.be}}}

\affil[1]{Institute for Theoretical Physics, ETH Z{\"u}rich}
\affil[2]{Instituut voor Theoretische Fysica, KU Leuven}
\affil[3]{Institute of Science and Technology Austria}

\date{\today}

\begin{document}

\maketitle

\begin{abstract}
We consider symmetry protected topological states of $2d$ quantum spin systems, with a finite symmetry group $G$.
It has been conjectured that such states are classified by the cohomology group $H^3(G,U(1))$, but the completeness of this classification is an open problem. We restrict ourselves to symmetry protected topological states that can be prepared from a pure product state by a symmetric entangler. For this class of states, we prove that the classification by $H^3(G,U(1))$ is complete. 
\end{abstract}

\tableofcontents

\section{Introduction}

The topology of gapped ground states of quantum lattice systems has grown into a well-developed subject.   The present paper focuses on states that are intrinsically trivial, i.e.\ deformable to product states, but that can fail to be deformable to product states if an on-site global symmetry is imposed. These are so-called symmetry protected topological (SPT) states \cite{guwen2009, Chen_2013,chenguwen2010, chen_gu_wen_2011, schuch2011MatrixProduct, pollman2012symmetry}. Moreover, since we are dealing with spin systems, the states are called \emph{bosonic}.   As is often done in the mathematical literature, we will not literally consider ground states of local Hamiltonians, but rather short-range entangled states, as this is a more direct approach that avoids irrelevant complications. 
More precisely,  we consider states $|\psi\rangle$ of a spin system on a $d$-dimensional lattice that are of the form
\begin{equation}\label{eq: spt state first}
   |\psi\rangle = C|\phi\rangle,
\end{equation}
where $\ket\phi$ is a pure product state and $C$ is a unitary implementing a finite-depth quantum circuit (FDQC), see Figure \ref{fig:intro-fdqc}.

    \begin{figure} [h]
    \centering
    \includegraphics[width=0.7\linewidth]{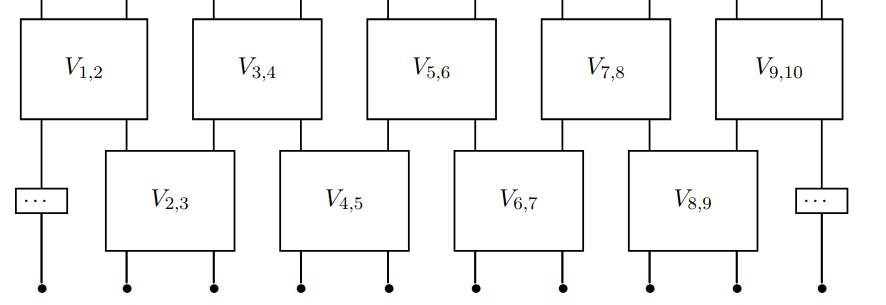}
    \caption{A finite-depth quantum circuit (FDQC) on a spin chain. Each $V_{i,i+1}$ is a unitary gate acting non-trivially on sites $i,i+1$. The dots represent a pure product state $\ket \phi = \otimes_{i \in \ZZ} \ket{\phi_i}$, where each $\ket{\phi_i}$ is a unit vector in a finite-dimensional on-site Hilbert space $\HH_i$.}
    \label{fig:intro-fdqc}
    \end{figure}

To bring symmetry into the setup, we assume that a unitary on-site representation $g \mapsto U_g$ of a finite group $G$ is given, and that $|\psi\rangle$ is invariant under this symmetry
\begin{equation}\label{eq: spt symmetry first}
|\psi\rangle \propto U_g |\psi\rangle, \qquad \forall g \in G.
\end{equation}

A state $|\psi\rangle$ satisfying \eqref{eq: spt state first} and \eqref{eq: spt symmetry first} is called a ($d$-dimensional) SPT state.
In reality, the above formulae featuring state vectors and unitary operators are not-well defined in infinite volume\footnote{One can save the day by considering a family of large but finite underlying lattices which increase towards $\bbZ^d$, while the range of the FDQC is not allowed to grow.}, and one should rather talk about states on observable $\rm C^*$-algebras and $\rm C^*$-algebra automorphisms, as we do in the present paper.  However, for the sake of recognisability, we ignore this complication in this introduction. 

We now proceed to the problem of classification of these SPT states. To do this, we introduce an equivalence relation on SPT states: $|\psi_1\rangle,|\psi_2\rangle$ are equivalent if and only if there exists an FDQC $S$ \emph{with symmetric gates}  such that $|\psi_2\rangle=S|\psi_1\rangle$. 
This approximates the physically relevant notion of \emph{adiabatic equivalence}, which is the same thing as asking equivalence under time evolutions generated by local and symmetric interactions \cite{hastings2005quasiadiabatic, hastings2004lieb, bachmann2012automorphic}. Such evolutions are well approximated by FDQCs with symmetric gates.

Given a pair of states $|\psi_1\rangle,|\psi_2\rangle$ on $d$-dimensional spin systems, we will also need the  \emph{stacked} state, which corresponds mathematically to taking the tensor product 
$$|\psi_1\rangle \otimes |\psi_2\rangle.$$
This also allows us to introduce the concept of \emph{stabilization}; we say two states $|\psi_1\rangle,|\psi_2\rangle$ are stably equivalent if and only if there are $G$-invariant pure product states $$|\phi_1\rangle,|\phi_2\rangle$$ whose on-site restrictions transform trivially under the on-site symmetry action, and such that the stacked states $|\psi_1\rangle \otimes |\phi_1\rangle$ and $|\psi_2\rangle \otimes |\phi_2\rangle$ are equivalent.  Note that stable equivalence allows to compare states with different on-site Hilbert spaces.
The set of $d$-dimensional  SPT states modulo stable equivalence, equipped with the stacking operation, is a \emph{monoid}. The problem of classification of SPTs is the problem of computing this monoid.

Since \cite{ogatatimereversal,schuch2011MatrixProduct, chen_gu_wen_2011, Chen_2013, Else_2014, duivenvoorden_topological_2013-1}, we know that this monoid contains an isomorphic copy of the cohomology group $H^{d+1}(G,U(1))$, with $d$ the spatial dimension of the underlying lattice (see also \cite{ogata2021,Freed2014,Freed2016,kapustin_symmetry_2014,kitaevtalk,wen.anomalies,Xiong_2018,Ragone:2024wie,tasaki_topological_2018,sopenkoindex}). It is known that for $d\geq 4$ the classification by group cohomology is not complete, see \cite{kapustin_symmetry_2014,fidkowski2020exactly}. It is believed\footnote{If anti-unitary symmetries are allowed, the group cohomology classification is not complete already in $d=3$, see \cite{kapustin_symmetry_2014, burnell2014,vishwanath}. } to be complete for $d < 4$.

By completeness of the classification, we mean that the full monoid is isomorphic to  the cohomology group $H^{d+1}(G,U(1))$.
For $d=0$, i.e.\ for a system without any spatial structure, the completeness of this classification is easily verified, and for $d=1$, it has been proven in \cite{ogata2019classification,schuch2011MatrixProduct,kapustin2021classification,de_o_carvalho_classification_2025}. {This leaves $d=2$ as one of the outstanding problems, and this is the subject of the present paper.}

We  prove that the classification is complete in $d = 2$ if one restricts to the class of SPT states with \emph{symmetric entanglers}\footnote{Throughout this paper \emph{symmetric entangler} means a symmetric FDQC which entangles the state. In the literature also symmetric entanglers that are non-trivial QCAs are considered. The beyond group cohomology states described in \cite{burnell2014, vishwanath} can be entangled by a symmetric QCA, but they do \emph{not} admit symmetric FDQC entanglers, see \cite[~Appendix C]{potter2018infinite}. One may therefore expect that the group cohomology classification is complete for SPTs with symmetric FDQC entanglers in $d=3$, even if anti-unitary symmetries are allowed. However, the beyond group cohomology state described in \cite{fidkowski2020exactly} does admit a symmetric FDQC entangler, so in $d=4$, the group cohomology classification is not complete even for the class of SPTs with symmetric entanglers.} \cite{zhang2023topological}.  We say that an SPT state $|\psi\rangle$ admits a symmetric entangler if, in the representation \eqref{eq: spt state first}, one can choose the unitary circuit $C$ and the pure product states $\phi$ to all be symmetric. That is, they  satisfy
\begin{equation}\label{eq: symmetric entangler first}
    U_g |\phi\rangle \propto |\phi\rangle, \qquad    [C,U_g]=0,\qquad \forall g\in G,
\end{equation}
in which case \eqref{eq: spt symmetry first} is automatically satisfied. To avoid confusion, note that \eqref{eq: symmetric entangler first} does \emph{not} imply that the FDQC implemented by $C$ consists of symmetric gates. Indeed, if the FDQC consisted of symmetric gates, the resulting SPT would simply be equivalent to a pure product state, hence trivial with respect to the classification.

At the time of writing, it is not clear to us whether the class of SPT states is  strictly larger than the class of SPT states with symmetric entanglers. For $d=0$ and $d=1$, these classes coincide, and for $d=2$, most authors seem   to believe that they coincide.  The present paper does not address this question.  Instead, and as already mentioned above, we prove (Theorem \ref{thm.main}) 

\begin{center}
\begin{tcolorbox}[
  width=0.9\linewidth,
  colback=white,
  colframe=gray,
  title=Classification of SPTs with symmetric entanglers,
  fonttitle=\bfseries,
]
 The monoid $\sSPT$ of stable equivalence classes of SPT states with symmetric entanglers in 2d is isomorphic to $H^3(G,U(1))$.
\end{tcolorbox}
\end{center}

The proof proceeds by first classifying the symmetric entanglers themselves.   Two symmetric entanglers $C_1,C_2$ are called equivalent if and only if there is a FDQC $S$ with symmetric gates such that 
$$
C_1 = S C_2.
$$
Since $S^{-1}$ is also an FDQC with symmetric gates, this relation is symmetric.    Two symmetric entanglers $C_1,C_2$ (i.e.\ FDQCs satisfying $[C_i,U_g]=0$) are called stably equivalent if and only if there are auxiliary spin systems, with respective identity operators $1_1,1_2$, such that $C_1\otimes 1_1$ is equivalent to $C_2\otimes 1_2$.  Note that, just as for SPT states, the symmetric entanglers form a monoid with respect to the stacking operation.

The complete invariant for this classification problem turns out to be the $H^3(G, U(1))$-valued anomaly, constructed previously for symmetric entanglers in \cite{Else_2016, zhang2023topological}. We use a slightly different definition of the same invariant \cite{bols2025classificationlocalitypreservingsymmetries} which expresses the index in terms of a locality preserving symmetry on a \emph{boundary algebra} \cite{freedman2020classification} associated to the entangler, which appears naturally in our analysis. We prove (Proposition \ref{prop.classification-of-symmetric-entanglers})

\begin{center}
\begin{tcolorbox}[
  width=0.9\linewidth,
  colback=white,
  colframe=gray,
  title=Classification of symmetric entanglers,
  fonttitle=\bfseries,
]
The monoid $\SE$ of stable equivalence classes of symmetric entanglers in 2d is isomorphic to $H^3(G,U(1))$.
\end{tcolorbox}
\end{center}

The corresponding result in $0d$ is again straightforward. The result in $1d$, which we review in Section \ref{sec.symmetric-entanglers-1d}, was proven in \cite{bols2021classification,gong2020classification}.

Let us now explain how this result implies the classification of SPTs with symmetric entanglers.
We will show that we can define an $H^3(G, U(1))$-valued index for an SPT state $\ket{\phi}$ with symmetric entangler $C$ by the index of the entangler, \ie $\ind \ket{\phi} := \ind \, C$. To show that this is well-defined, we must show that any two symmetric entanglers that entangle the same SPT state must have the same index. This is done in Lemma \ref{lem:well-defined index for G-states}. It follows that two SPT states $\ket{\psi_1}$ and $\ket{\psi_2}$ have the same index if and only if (after stacking) they can be entangled from the same pure product state $\ket{\phi_0}$ by symmetric entanglers $C_1$ and $C_2$ with the same index. But then the classification of symmetric entanglers implies that, after stacking, $C_1 = S C_2$ for an entangler $S$ with symmetric gates. We conclude that $\ket{\psi_1}$ and $\ket{\psi_2}$ are stably equivalent. This argument is carried out in Section \ref{sec.classification of 2d SPT with SE}. Altogether this shows that the index for SPT states with symmetric entanglers descends to a bijection from SPT phases with symmetric entanglers to $H^3(G, U(1))$, as required.

The proof of the classification of symmetric entanglers takes up the bulk of the present paper.  Also here, it is rather well-understood that there is a surjective homomorphism  $\SE \longrightarrow H^3(G, U(1))$ and the remaining challenge is to show its injectivity. 
In order to do this, we consider a symmetric entangler $C$ with trivial index, i.e.\ mapping to the trivial group cohomology class under this homomorphism, and show that $C$ is stably equivalent to the identity. To that end, we build a \emph{symmetric blend} of $C$ with the identity; a symmetric entangler that reduces to $C$ far to the left, and to identity far to the right of some vertical axis. This is done in Section \ref{sec: 2d classification}. The key step relies on the technique introduced in \cite{Else_2014}, whereby a symmetric entangler in $d$ dimensions gives rise to an anomalous symmetry action on a $(d-1)$-dimensional spin system that is the boundary of a $d$-dimensional spin system. In a recent paper \cite{bols2025classificationlocalitypreservingsymmetries}, see also \cite{seifnashri2025disentangling}, we completely classified $1$-dimensional anomalous symmetry actions and this makes the construction of our symmetric blend tick. 

Homotopy-theoretic approaches to invertible phases and lattice anomalies have recently been developed in several works. Kubota~\cite{kubota2025} gave an operator-algebraic realization of Kitaev's $\Omega$-spectrum of invertible phases, while Tu, Long and Else~\cite{else2026} developed a homotopy-theoretic framework for lattice anomalies. Czajka, Geiko and Thorngren~\cite{thorngren2025} construct $\Omega$-spectra of QCAs and FDQC-invertible states, leading to a generalized-cohomology classification of SPT phases up to blend equivalence.

Finally, here is an outline of our paper.
The set-up and the main theorem on the classification of 2d SPTs with symmetric entanglers are introduced in Section \ref{sec: setup}.  
In Section \ref{sec.symmetric-entanglers-1d}, we review the techniques used for the classification of 1d symmetric entanglers. These techniques and results will be needed for the classification of 2d symmetric entanglers. The classification of 2d symmetric entanglers is presented in Section \ref{sec: 2d classification}. 
In Section \ref{sec.classification of 2d SPT with SE}, we lift the index of 2d symmetric entangler to a well-defined index on 2d SPT phases and prove the main result.
In Appendix \ref{sec.cohomology}, we give a minimal introduction to group cohomology. In Appendix \ref{appendix.LPS}, we connect the classification results of LPS in the literature to our set-up. In particular we recall the definition of the index of LPS. Finally, in Appendix \ref{appendix.swindle}, we include the 0d classification and the classification of striped entanglers proved by the well-known Eilenberg-Mazur swindle argument in our setting for completeness. 

\section{Setup and Main Result}\label{sec: setup}

\subsection{Spin systems} \label{sec.spin-systems}

A $d$-dimensional spin system assigns to all sites $j \in \Z^d$ a Hilbert space $\HH_j \simeq \C^{d_j}$ with associated matrix algebra $\caA_j = \End(\HH_j)$. We assume that there is a $d_{\max}$ such that $d_j \leq d_{\max}$ for all $j \in \Z^d$. For finite $S \subset_f \Z^d$ we can form the algebra $\caA_S = \otimes_{j \in S} \caA_j$ of observables supported on $S$. If $S_1 \subset S_2$ are finite subsets of $\Z^d$ then there is a natural inclusion $\caA_{S_1} \hookrightarrow \caA_{S_2}$ given by tensoring with the identity of $\caA_{S_2 \setminus S_1}$. This makes the $\{ \caA_S \}_{S \subset_f \Z^d}$ into a directed system of $\rm C^*$-algebras whose limit $\caA$ is the quasi-local algebra of the spin system. For any $X \subset \Z^d$ (possibly infinite), the subalgebra $\caA_X \subset \caA$ of observables supported on $X$ is the direct limit of $\{ \caA_S \}_{S \subset_f X}$. We refer to standard references \cite{bratteliII,simon2014statistical,naaijkens2017quantum,brunoamandaI,nachtergaele.sims.ogata.2006} for details.

We emphasize that the assignment $\caA_j$ of on-site algebras is part of the data of what we mean by `spin system'. Later, we will encounter cases where a quasi-local algebra $\caA$ is cast as a spin system in two different ways, formally $\caA = \otimes_j \caA_j  = \otimes_j \caA'_j$ where the subalgebras $\caA_j, \caA_j'$ are not identical. Nevertheless, we will often denote a spin system by its quasi-local algebra when the assignment of on-site algebras is clear from context. If $d=1$, we will usually refer to spin systems as \textit{spin chains}.
For simplicity, we consistently use the terminology \emph{automorphism, isomorphism, etc.} for $^*$-automorphisms, $^*$-isomorphisms, etc. between $\rm C^*$-algebras. We denote by $\Aut{\AA}$ the automorphism group of $\AA$.

\subsection{Isomorphic spin systems} \label{sec:isomorphic-spin-systems}

We say that spin systems $\caA,
\caA'$ are isomorphic whenever their local dimensions agree, i.e.\ $d_j=d'_j$. In that case, we can pick unitaries $\mathscr F_j \in \mathcal U(\mathcal H_j, \mathcal H_{j}')$ and we can construct an isomorphism of $\rm C^*$-algebras $\mathscr F $:  For any finite $S \subseteq \ZZ^d$, we define the isomorphisms $\mathscr F_S:  \AA_S \to \AA'_S$ by setting 
$\mathscr F_S:= \otimes_{j\in S}\Ad{\mathscr F_j}$ (we use the notation $\Ad{\mathscr F_j}(A) = \mathscr F_jA\mathscr F_j^*$). These definitions are consistent and they extend to an isomorphism $\mathscr F : \AA\to \AA'$, which we refer to as a \emph{spin system isomorphism}. Given an automorphism $\alpha \in \Aut\AA$, we introduce the notation
\begin{equation}
    \iso\alpha{\mathscr F} := \mathscr F \circ \alpha \circ \mathscr F^{-1} \in \Aut{\AA'}. 
\end{equation}
Similarly, for states, we write
\begin{align}
    \iso\psi{\mathscr F} := \psi \circ \mathscr F^{-1},
\end{align}
which is a state on $\AA'$.

\subsubsection{Stacking of spin systems} \label{eq.stacking-spin-systems}

Given a pair of spin systems $\caA,
\caA'$, we consider the \emph{stacked} spin system 
$$\caA \otimes \caA',$$ 
whose on-site algebras are $\AA_j \otimes \AA_j'$, corresponding to the on-site Hilbert spaces  $\HH_j \otimes \HH'_j$.

\subsection{Quantum cellular automata and finite depth quantum circuits} \label{sec.qca}

For any $X \subset \Z^d$ we write $X^{(r)} := \{ j \in \Z^d \, : \, \dist(j, X) \leq r \}$
 for the \textit{$r$-fattening} of $X$.

A \emph{quantum cellular automaton} (QCA) on a spin system $\caA$ is an automorphism $\al \in \Aut{\AA}$ for which there exists $r \geq 0$ such that $\alpha(\caA_{X}) \subset  \caA_{X^{(r)}} $ for any $X\subset \Z^d$. The \emph{range} of a QCA is the smallest $r$ for which this holds. The inverse of a QCA of range $r$ is also a QCA of range $r$ (\cite[Lemma 3.1]{freedman2020classification}), a fact which we will use without comment throughout the paper. The quantum cellular automata on $\caA$ therefore form a subgroup of $\Aut{\caA}$ which we denote by $\QCA(\caA)$. The notation $(\alpha, \AA)$ is used to denote a QCA $\alpha$ acting on the spin system $\AA$. We denote by $(\id,\AA)$ the identity automorphism on $\caA$.

Adapting nomenclature from \cite{kapustin2025higher}, we say a {\emph{partitioned family of unitaries} (also called a \emph{partitioned unitary})} is a family $\sfV = \{ {V_i} \}_{i \in \caI}$  of unitaries, called \emph{gates}, whose supports are disjoint and have bounded diameters. The sets $\supp {V_i}$ are called \emph{blocks} of $\sfV$. The number $\max_{i \in \caI} \diam( \supp {V_i} )$ is called the \emph{range} of $\sfV$. To any partitioned unitary $\sfV$ we can associate a QCA $\gamma_{\sfV}$ by
\begin{equation} \label{eq.def-block-partitioned-QCA}
\gamma_{\sfV} = \bigotimes_{i \in \caI} \Ad{{V_i}}.
\end{equation}
Any QCA of this form is called a \emph{block partitioned QCA}. A \emph{depth $n$ circuit datum} is an $n$-tuple $\mathsf{C} = ( \sfV_{j} )_{j = 1}^n$ of partitioned unitaries. The component $\sfV_j$ is called the $j^{\rm{th}}$ layer of the circuit. The \emph{range} of $\mathsf{C}$ is the sum of the ranges of it layers. To any such circuit datum $\mathsf{C}$ we associate a QCA $\gamma_{\mathsf{C}} := \gamma_{\sfV_n} \circ \cdots \circ \gamma_{\sfV_1}$. Any QCA of this form is called a \emph{finite depth quantum circuit} (FDQC). The datum $\mathsf{C}$ is called a depth $n$ circuit representation for $\gamma_{\mathsf{C}}$. Note that the range of $\gamma_{\mathsf{C}}$ is bounded by the range of $\mathsf{C}$. Note also that circuit representations of FDQCs are not unique. That is, the map $\mathsf{C} \mapsto \gamma_{\mathsf{C}}$ is not injective.

Given two circuit data $\mathsf{C} = (\sfV_j)_{j = 1}^n$ and $\mathsf{C}' = ( \sfV'_j )_{j = 1}^m$ we define their composition as the circuit datum $\mathsf{C}' \circ \mathsf{C} := ( {\mathsf{W}}_j )_{j = 1}^{n+m}$ with ${\mathsf{W}}_j = \sfV_j$ for $j = 1, \cdots, n$ and ${\mathsf{W}}_j = \sfV'_{j-n}$ for $j = n+1, \cdots, n+m$. Then $\gamma_{\mathsf{C}'} \circ \gamma_{\mathsf{C}} = \gamma_{\mathsf{C}' \circ \mathsf{C}}$. For any circuit datum $\sfC = (\sfV_j)_{j = 1}^n$ we define the reversed circuit $\sfC^{\rev} := (\sfV^{\rev}_{n+1-j})_{j = 1}^n$ where for any partitioned unitary $\sfV = \{ {V_i} \}_{i \in \caI}$ we let $\sfV^{\rev} = \{ {V_i^{-1}} \}_{i \in \caI}$. Reversal is an involution on the set of circuit data. It preserves depth and range, and $\gamma_{\sfC^{\rev}} = \gamma_{\sfC}^{-1}$ for any circuit datum $\sfC$.

Given a circuit datum $\mathsf{C}$ and a subset $X \subset \Z^d$ we let $\mathsf{C}|_X$ be the circuit datum obtained from $\mathsf{C}$ by dropping all the gates whose support is not contained in $X$.

\subsection{Symmetries} \label{sec.symmetries}

Let $G$ be a finite group that is fixed throughout the paper, and let $e \in G$ denote the identity element. A \emph{locality preserving {$G$-}symmetry} (LPS) $(\beta, \AA)$ on $\caA$ is a group homomorphism $\beta : G \rightarrow \QCA(\caA)$. That is, for each $g \in G$ we have a quantum cellular automaton $\beta_g$ such that $\beta_e = \id$ and $\beta_g\beta_h = \beta_{gh}$ for all $g, h \in G$. The range of a locality preserving symmetry is the largest range of its component QCAs. {An LPS $\beta$ is called \emph{on-site} if
\[
\beta_g(\caA_j)=\caA_j,\qquad g\in G,\ j\in\Z^d.
\]
Equivalently, it factorizes over sites, or has range zero. Since every automorphism of a full matrix algebra is inner, such an LPS $\beta$ is necessarily of the form $\beta=\bigotimes_{j\in\Z^d}\Ad{U_j(g)}$ with each $g\mapsto U_j(g)$ a, possibly projective, representation of $G$.

An \emph{{on-site unitary representation datum}} is a specified family of genuine unitary representations
\[
\sfU=(U_j)_{j\in\Z^d},\qquad U_j:G\longrightarrow\caU(\HH_j).
\]
It induces the on-site LPS
\[
\beta^{\sfU}_g:=\bigotimes_{j\in\Z^d}\Ad{U_j(g)}.
\]
The pair $(\sfU,\caA)$ denotes a spin system equipped with this datum. An on-site LPS need not admit such a datum. Even when it does, its action does not determine the datum: multiplying each $U_j$ by a character of $G$ leaves $\beta^{\sfU}$ unchanged.

We write $\beta=\beta^{\sfU}$ when the datum is clear from the context, and use corresponding primes and decorations for the associated actions and data. For a finite set $S$, put $U_S(g)=\bigotimes_{j\in S}U_j(g)$. Stacking data means taking the sitewise tensor products $(\sfU\otimes\sfU')_j(g)=U_j(g)\otimes U'_j(g)$. Whenever an action is known to admit genuine on-site unitary implementers, we choose such a family when a datum is required.}

 Naturally, there is a refinement of the definition of isomorphic spin systems in the presence of symmetries: Spin systems $\AA, \AA'$ equipped with LPSs $\beta, \beta'$ are said to be \emph{equivariantly isomorphic} if the spin system isomorphism $\mathscr F: \AA \to \AA' $ defined in Section \ref{sec:isomorphic-spin-systems} intertwines the symmetry actions: $\beta_g' \circ \mathscr F = \mathscr F \circ \beta_g$ for all $g \in G$. We refer to such isomorphism as an \emph{equivariant spin system isomorphism}, denoted by $\mathscr F: (\beta, \AA) \to (\beta', \AA')$.

{For {on-site unitary representation data}, an \emph{equivariant spin system isomorphism}
$\mathscr F:(\sfU,\AA)\to(\sfU',\AA')$ is required to satisfy
\[
\mathscr F(U_j(g))=U'_j(g),\qquad j\in\Z^d,\ g\in G.
\]
This implies intertwining of the induced actions, but is a stronger requirement.}

We denote by $\Sym_{G,d}$ the set of LPSs on spin systems of dimension $d$.

\subsubsection{Stable equivalence of LPSs} \label{subsec:classification-lps}

Two LPSs $(\beta,\AA)$ and $(\beta',\AA')$ are \textit{equivalent} if there are
\begin{enumerate}
    \item a spin system isomorphism $\mathscr F: \AA \to \AA'$,
    \item an FDQC $(\gamma, \AA)$, 
\end{enumerate}
such that 
\begin{equation}
    \beta'_g = \iso{(\gamma^{-1} \circ \beta_g \circ \gamma)}{\mathscr F}
\end{equation}
holds for all $g\in G$. 
Locality preserving symmetries $(\beta,\AA)$ and $(\beta',\AA')$ can be \textit{stacked} into an LPS (recall Subsection \ref{eq.stacking-spin-systems}) $(\beta \otimes \beta', \AA \otimes \AA')$. Two LPSs $(\beta,\AA)$ and $(\beta',\AA')$ are \textit{stably equivalent}, denoted as
$$(\beta,\AA) \sim (\beta', \AA'),$$
if there exist {{on-site unitary representation data} $(\sfD,\tilde\AA)$ and $(\sfD',\tilde\AA')$, with induced actions $\delta=\beta^{\sfD}$ and $\delta'=\beta^{\sfD'}$,} such that $(\beta \otimes \delta, \AA \otimes \tilde \AA) \text{ and } (\beta' \otimes \delta', \AA' \otimes \tilde \AA')$
are equivalent. 

Stable equivalence is an equivalence relation on $\Sym_{G,d}$. The quotient $(\Sym_{G,d}/\sim)$ is an abelian monoid with multiplication given by stacking. The identity element is the equivalence class of {actions induced by {on-site unitary representation data}}.

\subsection{Symmetric entanglers}

A \emph{symmetric entangler} $(\alpha,\sfU,\AA)$ is an FDQC $(\alpha, \AA)$ that is \emph{$G$-equivariant} with respect to the action $\beta=\beta^{\sfU}$ of an {on-site unitary representation datum} $(\sfU, \caA)$, i.e.
$$
\alpha \circ \beta_g=  \beta_g \circ \alpha, \qquad \forall g \in G.
$$
{We denote by $\sEnt_{G,d}$ the set of all symmetric entanglers on $d$-dimensional spin systems}. A more restricted class of symmetric entanglers are \emph{FDQCs with symmetric gates}. More precisely, $(\alpha, \sfU, \AA)$ is an FDQC with symmetric gates if $\alpha$ has a circuit representation each of whose layers has gates $\{V_i\}$ that satisfy $\beta_g(V_i) = V_i$ for all $g \in G$. FDQCs with symmetric gates on a given $(\sfU, \caA)$ form a normal subgroup of the group of symmetric entanglers on $(\sfU, \caA)$ with composition as the group operation. That is, if $(\al, \sfU, \caA)$ is a symmetric entangler and $(\gamma, \sfU, \caA)$ is an FDQC with symmetric gates, then $\al \circ \gamma \circ \al^{-1}$ is again an FDQC with symmetric gates.

\subsection{Classification of symmetric entanglers}

We say two symmetric entanglers $(\alpha, \sfU,\AA),\ (\alpha',\,\sfU',\AA')$ are equivalent if there exist
\begin{enumerate}
    \item an equivariant spin system isomorphism $\mathscr F: (\sfU, \AA) \to (\sfU', \AA')$, 
    \item an FDQC $(\gamma,\sfU,\AA)$ with symmetric gates
\end{enumerate}
such that 
\begin{equation} 
\alpha' = \iso{(\alpha \circ \gamma)}{\mathscr F}.
\end{equation}
This equivalence relation is symmetric, because it implies that $\alpha  = \iso{(\alpha' \circ \iso{(\gamma^{-1})}{\mathscr F})}{\mathscr F^{-1}}$, and $\iso{(\gamma^{-1})}{\mathscr F}$ is also an FDQC with symmetric gates. 

Two symmetric entanglers $(\alpha,\sfU,\AA)$ and $(\alpha',\sfU',\AA')$ can be \textit{stacked} into a new symmetric entangler (recall Subsection \ref{eq.stacking-spin-systems})
$(\alpha \otimes \alpha', \sfU \otimes \sfU',\AA \otimes \AA').$
Two symmetric entanglers $(\alpha, \sfU,\AA)$ and $(\alpha',\sfU',\AA')$ are \emph{stably equivalent}
$$(\alpha, \sfU,\AA) \sim (\alpha',\sfU',\AA'),$$
if there are {on-site unitary representation data} $(\tilde \sfU,\tilde \AA)$ and $(\tilde\sfU',\tilde\AA')$ such that $(\alpha \otimes \id, \sfU \otimes \tilde \sfU, \AA \otimes \tilde \AA)$ and $(\alpha' \otimes \id, \sfU' \otimes \tilde \sfU', \AA' \otimes \tilde \AA')$ are equivalent. We denote the set of stable equivalence classes by
$$ \SE_{G, d} := \sEnt_{G, d} / \sim. $$

The stacking operation makes the sets $\sEnt_{G,d}$ and  $\SE_{G, d}$
into monoids. One of the main results of the present paper is the identification of $\SE_{G,2}$ with the third cohomology group of $G$ with $U(1)$ coefficients. We write  $H^{n}(G,U(1))$ for the $n^{\mathrm{th}}$ cohomology group of $G$ with coefficients in $U(1)$, see  Appendix \ref{sec.cohomology}. Then 
\begin{proposition} [Classification of symmetric entanglers] \label{prop.classification-of-symmetric-entanglers}
For $d = 0, 1, 2$, the monoid $\SE_{G, d}$ is a group and there is a group isomorphism
    \begin{equation}
       \ind :  \SE_{G, d} \,  \simeq \,  H^{d+1}(G,U(1)).
    \end{equation}
\end{proposition}

The case $d = 0$ is elementary, see Lemma \ref{lem.classification of 0d symmetric entanglers} in Appendix \ref{appendix.swindle}. The case $d=1$ has been proven by \cite{bols2021classification,gong2020classification}, and we review the proof in Section \ref{sec.symmetric-entanglers-1d}. The proof of the Proposition for $d=2$ is the subject of Section \ref{sec: 2d classification}  and it constitutes the bulk of the present paper. 
The construction of a group homomorphism $ \SE_{G, d} \,  \to \,  H^{3}(G,U(1))$ can be found in the literature, and its surjectivity was at least suggested. The injectivity is the novel part.

\begin{remark}
    The index $\ind : \SE_{G, d} \rightarrow H^{d+1}(G, U(1))$ will be defined in each case $d = 0, 1, 2$ as the quotient of an index map $\ind : \sEnt_{G, d} \rightarrow H^{d+1}(G, U(1))$ defined on individual symmetric entanglers. Lemma \ref{lem:composition is equivalent to stacking} together with Proposition \ref{prop.classification-of-symmetric-entanglers}, shows that for $d = 0, 1, 2$ the index 
    $$ \ind : \sEnt_{G, d}(\caA, \sfU) \rightarrow H^{d+1}(G, U(1))$$
    is a group homomorphism, where $\sEnt_{G, d}(\caA, \sfU)$ is the group of symmetric entanglers on a fixed spin system $\caA$ with fixed {on-site unitary representation datum} $\sfU$, with group operation given by composition.
\end{remark}

The following Lemma says that the stack of two symmetric entanglers is stably equivalent to their composition. A similar argument was used in \cite[Theorem 2.3]{freedman2020classification}, \cite{freedman2022group} to show that the group of (stable) equivalence classes of QCAs is always abelian.

\begin{lemma} \label{lem:composition is equivalent to stacking}
    Let $(\al_1, \sfU, \caA)$ and $(\al_2, \sfU ,\caA)$ be symmetric entanglers on a $d$-dimensional spin system $\caA$. Then  
    $$(\al_2 \circ \al_1, \sfU, \caA) \sim ({\al_2 \otimes \al_1}, \sfU \otimes \sfU, \caA \otimes \caA).$$
\end{lemma}

\begin{proof}
    For each site $j \in \Z^d$, let $\Omega_j \in \caU(\HH_j \otimes \HH_j)$ be the swap unitary, uniquely determined up to phase by
    the property that $\Omega_j (A \otimes B) \Omega_j^* = B \otimes A $ for all $A, B \in \caA_j$.

    We have $(\beta_g \otimes \beta_g)(\Omega_j) = \Omega_j$ for all $j \in \Z^d$ and all $g \in G$ so
    $$ \big( \gamma := \otimes_{j \in \Z^d} \Ad{\Omega_j}, \sfU \otimes \sfU, \caA \otimes \caA \big) $$
    is an FDQC with symmetric gates. 
    We thus have:
    \begin{align*} 
        \al_2 \circ \al_1 &\sim (\al_2 \circ \al_1 )\otimes \id = (\al_2 \otimes \id) \circ (\al_1 \otimes \id) = (\al_2 \otimes \id) \circ \gamma \circ (\id \otimes \al_1) \circ \gamma^{-1} \\
        &= (\al_2 \otimes \al_1) \circ (\id \otimes \al_1)^{-1} \circ \gamma \circ (\id \otimes \al_1) \circ \gamma^{-1} \sim \al_2 \otimes \al_1
    \end{align*}
    because $(\id \otimes \al_1)^{-1} \circ \gamma \circ (\id \otimes \al_1) \circ \gamma^{-1}$ is an FDQC with symmetric gates. 
\end{proof}

\subsection{Symmetry protected states} \label{sec.SPTs}

A \emph{symmetry protected topological (SPT) phase} is a stable equivalence class of symmetric short-range entangled states, which we define now. 

We say a state $\psi$ on a spin system $\AA$ is \emph{$G$-invariant} with respect to an {on-site unitary representation datum} $(\sfU,\AA)$ if $\psi \circ \beta^{\sfU}_g = \psi$ for all $g \in G$. For simplicity, when this condition is met for some fixed on-site symmetry that is clear from the context, we will say that $\psi$ is $G$-invariant.

\subsubsection{Symmetric short-range entangled states}

A \emph{product state} $\phi$ on a spin system $\AA$ is a state that satisfies 
$$\phi (AB) = \phi(A) \phi(B)$$
for any $A,B \in \AA$ supported on disjoint subsets of $\ZZ^d$. Any pure product state $\phi$ on $\caA$ restricts to pure states on the on-site algebras $\caA_j$ represented by unit vectors $\ket{\phi_j} \in \HH_j$, so that $\phi(\cdot ) = \otimes_j \langle \phi_j |\cdot| \phi_j\rangle$. If $\phi$ is moreover $G$-invariant under $\beta=\beta^{\sfU}$ then  $U_j(g)\ket{\phi_j} = q_j(g)\ket{\phi_j}$ for on-site \emph{$G$-charges} $q_j \in \rm{Hom}(G,U(1))$. A \emph{special} $G$-invariant product state is a $G$-invariant pure product state satisfying $q_j \equiv 1$ for every $j \in \Z^d$. 

We restrict ourselves to short-range entangled symmetric states that can be prepared by a symmetric entangler from an arbitrary $G$-invariant pure product state $\phi$ on $\AA$. 
That is,
$$\psi_\alpha := \phi \circ \alpha$$
for a symmetric entangler $(\alpha,\sfU,\AA)$. Any triple $(\psi_\alpha,\sfU,\AA)$ with $\psi_{\al}$ of this form is called a \emph{$G$-state}. If $\psi_\alpha$ is moreover a (special) pure product state then we call $(\psi_\alpha,\sfU,\AA)$ a (special) product $G$-state.

We denote by $\sSta_{G,d}$ the set of $G$-states in dimension $d$.

\subsubsection{Stable equivalence of $G$-states}

Two $G$-states $(\psi_\alpha, \sfU, \AA)$ and $(\psi_{\alpha'}, \sfU',\AA')$ are \emph{equivalent} if 
\begin{enumerate}
    \item there exists an equivariant spin system isomorphism $\mathscr F: (\sfU,\AA) \to (\sfU',\AA')$,
    \item there exists an FDQC $(\gamma,\sfU,\AA)$ with symmetric gates,
\end{enumerate}
such that 
\begin{equation}
 \psi_{\alpha'} = \iso{(\psi_\alpha \circ \gamma)}{\mathscr F}.
\end{equation}
Similarly, two $G$-states $(\psi_{\alpha'}, \sfU', \AA')$ and $(\psi_{\alpha}, \sfU,\AA)$ are \emph{stably equivalent},
$$(\psi_\alpha,\sfU,\AA) \sim (\psi_{\alpha'}, \sfU', \AA'),$$
if there exist special product $G$-states $(\phi, \tilde \sfU, \tilde \AA )$ and $(\phi', \tilde \sfU', \tilde \AA')$ such that $(\psi_\alpha \otimes \phi, \sfU \otimes \tilde \sfU, \AA \otimes \tilde \AA)$ and $(\psi_{\alpha'} \otimes \phi', \sfU' \otimes \tilde \sfU', \AA' \otimes \tilde \AA')$ are equivalent. Stable equivalence is an equivalence relation on the set $\sSta_{G,d}$.

\subsubsection{Classification of $G$-states}

A \emph{symmetry protected topological (SPT) phase with symmetric entanglers} in dimension $d$ is a stable equivalence class of $G$-states. We denote the set of such phases by
$$\sSPT_{G,d} := (\sSta_{G,d}/\sim).$$
It has the structure of a commutative monoid with product operation given by stacking. It is a submonoid of $\mathsf{SPT_{G,d}}$, the monoid of symmetry protected topological phases, consisting of stable equivalence classes of symmetric short-range entangled states in $d$ dimensions. In dimension $d=1$, the monoid $\mathsf{SPT}_{G,1}$ has been completely classified \cite{ogata2019classification,schuch2011MatrixProduct,kapustin2021classification,de_o_carvalho_classification_2025}.

The following Theorem, whose proof can be found in Section \ref{sec.classification of 2d SPT with SE}, gives a complete classification of 2d SPTs with symmetric entanglers:

\begin{theorem} [Classification of 2d SPTs with symmetric entanglers] \label{thm.main}
    There is a group isomorphism
    \begin{equation}
        \ind : \sSPT_{G,2}  \rightarrow H^{3}(G,U(1))
    \end{equation}
    which assigns to the class $[\psi_{\al}]$ the index $\ind \, \al$ of the symmetric entangler of any representative $G$-state.
\end{theorem}
It remains an open question whether all 2d symmetric short-range entangled states can be prepared from a pure product state by a symmetric entangler.

\section{Classification of 1d symmetric entanglers} \label{sec.symmetric-entanglers-1d}

We review the methods used for classification of 1d symmetric entanglers. These results will be needed for and serve as a guideline for the classification of 2d symmetric entanglers in the next section. 

{Equivariant QCAs on spin chains have two classification invariants: the GNVW index \cite{Gross_2012} and an $H^2(G,U(1))$-valued invariant \cite{gong2020classification, bols2021classification}. For symmetric entanglers, which are FDQCs, the GNVW index is trivial. We prove below that the remaining invariant completely classifies symmetric entanglers up to the stable equivalence defined in Section~\ref{sec: setup}.}

\begin{lemma} \label{lem.classification-1d-symmetric-entanglers}
     The monoid $\SE_{G, 1}$ is a group, and there is a group isomorphism 
    \begin{equation}
        \SE_{G, 1} \simeq H^2(G,U(1)).
    \end{equation}
\end{lemma}

This Lemma is the $1$-dimensional case of Proposition \ref{prop.classification-of-symmetric-entanglers}, it follows immediately from Lemmas \ref{lem.index-1d-symmetric-entanglers}, \ref{lem.surjectivity-1d}, and \ref{lem:injectiveness-1d} below.

\subsection{An index for 1d symmetric entanglers}

Let $(\alpha, \sfU, \AA)$ be a symmetric entangler in $d=1$.

\subsubsection{The boundary algebra} \label{sec.boundary-algebra-1d}

Fix a circuit representation $\mathsf{A}$ for $\alpha$ and let $r\geq nq$ where $n$ is the depth of the circuit and $q$ is the maximum range of its layers.  Then $\range(\al) \leq \range(\mathsf{A}) \leq r$.

Define half-lines $\Left{x} := \{ i \in \ZZ\ |\ i < x \}$, and boundary regions
\begin{equation} \label{eq.boundary-region-1d}
    \partial_r \Left{x} := \{ i\in \ZZ\ |\ x-r \le i < x\}.
\end{equation}
The corresponding \emph{boundary algebra} $\mathcal P^{[r]}_{x}$ associated to $\alpha$ is defined as (cf. \cite{freedman2020classification, Haah2023})
\begin{equation} \label{eq.boundary-algebra-hastings-freedman-1d}
   \mathcal P^{[r]}_{x} :=  \alpha (\AA_{\Left x}) \cap \AA_{\Left{x-r}}'
\end{equation}
where the prime denotes the commutant in $\AA$. This object clearly depends on $\alpha$, but we choose to not indicate this dependence in the notation. 

We now give an alternative construction of the boundary algebra $\mathcal P^{[r]}_{x}$ in terms of the circuit representation $\mathsf{A}$. Define $(\eta_x,\AA)$ to be the FDQC associated to the circuit datum $\mathsf{B}$ defined as follows:

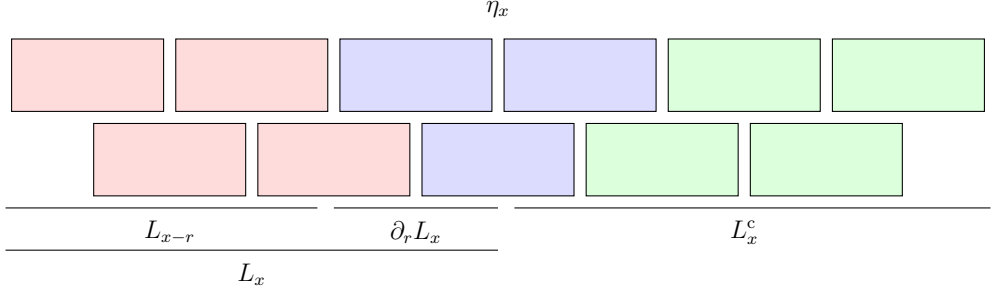
\begin{figure}
 \centering 
 \begin{tikzpicture}[scale=0.8, transform shape]
    \def\splitgap{0.2}  
    \def\xdist{2.7}
    \def\boxwidth{2.5}
    \def\boxheight{1.2}
    \def\ydist{1.4}
    \def\greenheight{3.1}
    \def\greenwidth{0.35} 
    \def\greendist{2.55}  

    \definecolor{softred}{RGB}{255,220,220}
    \definecolor{softblue}{RGB}{220,220,255}
    \definecolor{softgreen}{RGB}{220,255,220}

    \foreach \j [evaluate=\j as \label using int(\j)] in {-4,-2} {
        \pgfmathsetmacro{\x}{(\j + 6)/2 * \xdist}
        \node[draw=black, fill=softred, text=black, minimum width=\boxwidth cm, minimum height=\boxheight cm] 
            at (\x, 0){} ;
    }

    \foreach \j [evaluate=\j as \label using int(\j)] in {-5,-3} {
        \pgfmathsetmacro{\x}{(\j + 6)/2 * \xdist}
        \node[draw=black, fill=softred, text=black, minimum width=\boxwidth cm, minimum height=\boxheight cm] 
            at (\x, \ydist) {};
    }
    
    \foreach \j [evaluate=\j as \label using int(\j)] in {0} {
        \pgfmathsetmacro{\x}{(\j + 6)/2 * \xdist}
        \node[draw=black, fill=softblue, text=black, minimum width=\boxwidth cm, minimum height=\boxheight cm] 
            at (\x, 0) {};
    }

    \foreach \j [evaluate=\j as \label using int(\j)] in {-1,1} {
        \pgfmathsetmacro{\x}{(\j + 6)/2 * \xdist}
        \node[draw=black, fill=softblue, text=black, minimum width=\boxwidth cm, minimum height=\boxheight cm] 
            at (\x, \ydist) {};
    }

    \foreach \j [evaluate=\j as \label using int(\j)] in {2,4} {
        \pgfmathsetmacro{\x}{(\j + 6)/2 * \xdist}
        \node[draw=black, fill=softgreen, text=black, minimum width=\boxwidth cm, minimum height=\boxheight cm] 
            at (\x, 0) {};
    }

    \foreach \j [evaluate=\j as \label using int(\j)] in {3,5} {
        \pgfmathsetmacro{\x}{(\j + 6)/2 * \xdist}
        \node[draw=black, fill=softgreen, text=black, minimum width=\boxwidth cm, minimum height=\boxheight cm] 
            at (\x, \ydist) {};
    }

\node[] at (3*\xdist, 2.5) {$\eta_x$};
\draw (0 , -0.8) -- (1.9*\xdist, -0.8);
\node[] at (\xdist,-1.2) {$\Left{x-r}$};
\draw (2 * \xdist, -0.8) -- (3*\xdist, -0.8);
\node[] at (2.5*\xdist,-1.2) {$\partial_r L_x$};
\draw (3.1 * \xdist, -0.8) -- (6*\xdist, -0.8);
\node[] at (4.5*\xdist,-1.2) {$L_x^{\text{c}}$};
\draw (0* \xdist, -1.5) -- (3*\xdist, -1.5);
\node[] at (1.5*\xdist,-1.9) {$\Left x$};
\end{tikzpicture}
    \caption{An illustration of a depth 2 circuit representation of $\al$. The gates comprising $\eta_x$ are shown in blue.}
    \label{fig:eta1d}
\end{figure}
\begin{itemize}
    \item From the first layer of $\mathsf{A}$, we retain gates which connect the half-line with its complement, i.e. whose support intersects both $\Left x$ and $\Left{x}^c$. These gates form the first layer of $\mathsf{B}$. Let $B_1$ be the union of the supports of these gates ($B_1=\emptyset$ if there is no such gate).

    \item For $i=2,\ldots,n$, the $i^{\mathrm{th}}$ layer of $\mathsf{B}$ consists of all gates in the $i^{\rm{th}}$ layer of $\mathsf{A}$ whose support intersects $B_{i-1}$ or connects the half-line to its complement. 
    We call the union of $B_{i-1}$ with the supports of these new gates $B_{i}$. 
\end{itemize}

This construction is illustrated by Figure \ref{fig:eta1d}, it yields
\begin{equation} \label{eq.boundary-algebra-eta-1d}
\alpha(\caA_{\Left{x}})=\eta_x(\caA_{\Left x})= \eta_x(\caA_{\partial_r \Left x}) \otimes \caA_{\Left{x-r}} = \mathcal P^{[r]}_{x}  \otimes \caA_{\Left{x-r}},
\end{equation}
where first and second equalities follow by construction of $\eta_x$ and the last equality follows from the definition \eqref{eq.boundary-algebra-hastings-freedman-1d}.


\begin{lemma}
\label{lem.invariant-boundary}
        The boundary algebra $\mathcal P^{[r]}_{x}$ is $G$-invariant, i.e. $\beta_g(\mathcal P^{[r]}_{x})=\mathcal P^{[r]}_{x}$ for all $g \in G$.

\end{lemma}
\begin{proof} 
        Since the action $\beta$ is on-site, $\AA_{\Left{x-r}}' = \AA_{ L_{x-r}^c }$  is $G$-invariant. Moreover, $[\alpha,\beta_g]=0$, and hence  $\alpha (\AA_{\Left x}) $ is also $G$-invariant. 
\end{proof}

The boundary algebra $\mathcal P^{[r]}_{x}$ is therefore isomorphic to a symmetric finite-dimensional full matrix algebra $\eta_x(\AA_{\partial_r \Left x})$ contained in the finite-dimensional algebra $\AA_{[x-r,x+r]}$. Every automorphism of a full matrix algebra is inner, so there is a collection of unitaries $\{W_{x}^{[r]}(g)\}_{g \in G} \subset \mathcal P^{[r]}_{x}$ such that
\begin{equation} \label{eq:1d boundary action}
    \beta_g (\mathcal P^{[r]}_{x}) = \Ad{W_{x}^{[r]}(g)} (\mathcal P^{[r]}_{x}), \qquad \qquad W_x^{[r]}(g) W_x^{[r]}(h) = \mu_x^{[r]}(g,h) W_x^{[r]}(gh),
\end{equation}
for a 2-cocycle $\mu_x^{[r]} \in Z^2(G,U(1))$ (see Appendix \ref{appendix.projective-reps} for a discussion on projective representations). The unitaries $W_x^{[r]}(g)$ are unique up to a choice of $U(1)$ phases. We define the \emph{index} of $(\alpha,\sfU,\AA)$ as
\begin{equation} \label{eq.index1d}
    \ind(\alpha) = [\mu_x^{[r]}] \in H^2(G,U(1)).
\end{equation}

\begin{lemma} \label{lem.index-independent-of-r-1d}
    The index map \eqref{eq.index1d} is well defined. That is, the class $[\mu_x^{[r]}]$
    \begin{enumerate}
    \item does not depend on the choice of $r \geq nq$ made when defining the boundary \eqref{eq.boundary-region-1d};
    \item does not depend on the circuit representation $\mathsf{A}$ of $\alpha$;
    \item does not depend on $x \in \ZZ$.
    \end{enumerate}
\end{lemma}
\begin{proof} \hfill
\begin{enumerate}
    \item Let $r_0$ be a positive integer. Since $\eta_x|_{\AA_{[x-r-r_0,x-r-1]}} = \id$, it follows that 
        $$\mathcal P^{[r+r_0]}_{x} = \mathcal P^{[r]}_{x} \otimes \AA_{[x-(r+r_0), x-r-1]}.$$
        
        Consequently, 
        $$W_{x}^{[r+r_0]} (g) = c(g) W_x^{[r]} (g) \otimes U_{[x-r-r_0,x-r-1]} (g)$$
        for phases $c(g) \in U(1)$. Since $g \to U_i(g)$ are linear representations, we get $[\mu_{x}^{[r+r_0]}] = [\mu_x^{[r]}]$. 

        \item The unitaries $W_x^{[r]}$ are, up to $U(1)$ phases, uniquely determined by the symmetry action on the boundary algebra $\caP_x^{[r]}$. Moreover, the boundary algebra, by definition, depends on the circuit representation of $\alpha$ only through the constraint $r \geq nq$. By item 1, the index does not depend on this choice.

        \item Let $x_0,r_0$ be positive integers such that $r_0>2r$ and $x_0<r$. We compute
        \begin{align*}
            \caP_{x+x_0}^{[r_0]} &= \alpha(\AA_{\Left {x+x_0}}) \cap \AA_{\Left{x+x_0-r_0}}' \\ 
            &= \left[ \alpha(\AA_{\Left {x}}) \otimes \alpha (\AA_{\Left{x+x_0} \setminus \Left{x}})  \right]\cap \AA_{\Left{x+x_0-r_0}}'
            \intertext{by simply using the definition. We proceed by using the commutativity between algebras $\alpha(\AA_{\Left{x+x_0} \setminus \Left {x}})$ and $\AA_{\Left{x+x_0-r_0}}$, which follows from the assumptions on $x_0$ and $r_0$. Thus we can continue:}
            &= \left[ \alpha(\AA_{\Left {x}}) \cap  \AA_{\Left{x+x_0-r_0}}'  \right] \otimes \alpha (\AA_{\Left{x+x_0} \setminus \Left{x}}) \\
            &= \caP_{x}^{[r_0-x_0]} \otimes \alpha (\AA_{\Left{x+x_0} \setminus \Left{x}}).
        \end{align*}
        Since $\beta_g \circ \alpha(\AA_{\Left {x+x_0} \setminus \Left{x}}) = \alpha \circ \Ad{U_{[x,x+x_0)}(g)} (\AA_{\Left {x+x_0} \setminus \Left{x}})$ and $g\to U_{[x,x+x_0)}(g)$ is a linear representation, we obtain 
        $$[\mu_{x+x_0}^{[r_0]}] = [\mu_{x}^{[r_0-x_0]}].$$
        The claim then follows from item 1 above.
\end{enumerate} 
\end{proof}

\begin{lemma} \label{lem.index-1d-symmetric-entanglers}
    The index \eqref{eq.index1d} descends to a homomorphism of monoids
    \begin{equation*}
        \ind: \SE_{G, 1} \to H^2(G,U(1))
    \end{equation*}
    which we call the index map.
\end{lemma}

\begin{proof}
    Fix a symmetric entangler $(\alpha, \sfU, \AA)$.
    \begin{enumerate}
          
        \item \emph{The index map is well defined, that is, if two symmetric entanglers are stably equivalent, their indices coincide.}\\[1mm]  
        Firstly, given an equivariant spin system isomorphism $\mathscr F: (\sfU, \AA) \to (\sfU',\AA')$, the analogous defining unitaries \eqref{eq:1d boundary action} for the symmetric entangler $(\iso{\alpha}{\mathscr F}, \sfU', \AA')$ are $\{ \mathscr F(W_{x}^{[r]}(g)) \}_{g\in G}$. Thus the 2-cocycle $\mu_x^{[r]}$ remains invariant. 
Secondly, we prove that the index is invariant w.r.t.~composition with FDQCs with symmetric gates. Notice that the index is manifestly locally computable: it depends on the projective action $g \to W_x^{[r]} (g)$ on $\caP_x^{[r]}$, which depends only on the local action of $\alpha$ on $\AA_{\partial_r \Left x}$. Let $(\gamma, \sfU,\AA)$ be an FDQC with symmetric gates, and denote by $(\gamma_{R{x}}, \sfU, \AA)$ an FDQC with symmetric gates obtained from $\gamma$ by dropping gates from a circuit representation for $\gamma$ whose supports intersect $\Left x$. By local computability and independence of half-lines we obtain
        \begin{align*}
            \ind(\alpha) = \ind(\alpha \circ \gamma_{R_{x}}) = \ind(\alpha \circ \gamma). 
        \end{align*}

        \item \emph{The index map is a homomorphism.}\\[1mm] If $\alpha = \id$, then $\caP_x^{[r]} = \caA_{[x-r, x-1]}$ and $W_x^{[r]}(g) = U_{[x- r, x-1]}(g)$. Triviality of its index follows from {the fact that each component $U_i$ of the {on-site unitary representation datum} $\sfU$ is a genuine unitary representation of $G$}. Moreover, if two given symmetric entanglers $\al$ and $\tilde \al$ yield boundary algebras $\caP_x^{[r]}$ and $\tilde \caP_x^{[r]}$ with symmetry implemented by unitaries $\{W_x^{[r]} (g)\} \subset \caP_x^{[r]}$ and $\{\tilde W_x^{[r]} (g)\} \subset \tilde \caP_{x}^{[r]}$, then the stacked symmetric entangler yields boundary algebra $\caP_x^{[r]} \otimes \tilde \caP_x^{[r]}$ with symmetry implemented by $\{ W_x^{[r]} (g) \otimes \tilde W_x^{[r]} (g) \}$. It follows that the index is additive with respect to stacking.
    \end{enumerate}

\end{proof}

\subsection{Surjectivity of the index map in 1d}

\begin{lemma} [Surjectivity of the index map in 1d] \label{lem.surjectivity-1d} 
 Given a class $[\mu] \in H^2(G,U(1))$, there exists a 1d symmetric entangler $(\alpha, \sfU,\AA)$ such that $\ind(\alpha) = [\mu]$. 
\end{lemma}

\begin{proof}
    For any 2-cocycle $[\mu]$, let $(V_{\mu}, \mathcal K^{\mu})$ be a projective unitary representation $g \mapsto V_\mu(g)$ of $G$ on a finite dimensional Hilbert space $\KK^{\mu}$ with $V_{\mu}(g) V_{\mu}(h) = \mu(g, h) V_{\mu}(gh)$.
    
    Let $\caA^{\mu}$ be the spin chain with on-site degrees of freedom $\caA_j^{\mu} \simeq \End( \mathcal K^{\mu} )$. For each $j \in \Z$ we then have a projective representation $V_{j, \mu} : G \rightarrow \caU( \caA_j^{\mu})$ whose simultaneous action defines $\beta^{\mu}_g := \otimes_j \Ad{V_{j, \mu}(g)}$ for each $g \in G$.  
    
    Consider the stacked chain
    $ \caA := \caA^{\mu} \otimes \caA^{\bar\mu} \otimes \caA^{\mu}$ with on-site symmetry $\beta = \beta^{\mu} \otimes \beta^{\bar \mu} \otimes \id$. {This action admits an {on-site unitary representation datum} because it is implemented by} $U_i=V_{i,\mu} \otimes V_{i,\bar \mu} \otimes \I_i$, which is a unitary representation of $G$. {Set $\sfU=(U_i)_i$.} Let $U_{[x,y]}$ be the product $\prod_{x\leq i \leq y} U_i$. 

    Let $\al = S \otimes \id \otimes S^*$ where $S$ is the right shift on the spin chain $\caA^{\mu}$. One easily verifies that $\al \circ \beta_g = \beta_g \circ \al$ for all $g \in G$, and since the GNVW index of $\al$ is trivial \cite{Gross_2012}, it is an FDQC of range $R(\alpha)=1$, and $nq = 4$. Therefore, $(\al, \sfU, \caA)$ is a symmetric entangler.

    Let us now compute its index. Let us write $\caA_i = \caA_{i, 1} \otimes \caA_{i, 2} \otimes \caA_{i, 3}$ for the decomposition of each on-site degree of freedom of $\caA = \caA^{\mu} \otimes \caA^{\bar \mu} \otimes \caA^{\mu}$. With $L_x = \{ i \in \Z \, : \, i < x\}$, as introduced before, we have 
    $$\al(\caA_{L_0}) = \caA_{L_{-1}} \otimes \left(\caA_{-1, 1}\otimes \caA_{-1, 2}\right)\otimes \caA_{0, 1}.$$
    With the choice $r = 4$ the boundary algebra is $\caP_{0}^{[r]} = \caA_{[-4, -2]} \otimes \caA_{-1, 1} \otimes \caA_{-1, 2} \otimes \caA_{0, 1}$. The symmetry $\beta$ restricts to $\caP_{0}^{[r]}$ as the projective representation $U_{[-4, -2]} \otimes V_{\bar \mu} \otimes V_{\mu} \otimes V_{\mu}$. We conclude that $\ind(\al) = [\mu]$.
\end{proof}

\subsection{Symmetric blending in 1d} \label{sec.symmetric-blending-1d}

Let $(\alpha,\sfU,\AA)$ be a symmetric entangler in $d=1$ such that $\ind(\alpha) = [1]$, i.e. the representations $g \to W_x^{[r]}(g)$ of Eq. \eqref{eq:1d boundary action} are linear. Fix a circuit representation $\mathsf{A}$ for $\alpha$ and let $r\geq nq$ where $n$ is the depth of the circuit and $q$ is the maximum range of its layers.

\subsubsection{Auxiliary spin chain}

Let $ \caC'$ be a spin chain with on-site degrees of freedom
\begin{align*}
    \caC'_{x-1} \simeq \begin{cases} \caP_x^{'[r]} \otimes \End \mathbb C[G], \qquad &x \in 11r\ZZ, \\
    \mathbb C, & \text{otherwise,}
    \end{cases}
\end{align*}
where $\caP^{'[r]}_x \simeq \caP_x^{[r]}$, and $\End \mathbb C[G]$ is the matrix algebra over $\mathbb C[G]$. 

For each $x \in 11r \Z$, fix an isomorphism
\begin{align} \label{eq.pi-1d}
\pi_x : \caP_x^{[r]} \to \caP^{'[r]}_x
\end{align}
between the boundary algebra at $x$ and its copy. 

We provide the chain $\caC'$ with an on-site symmetry $\{\beta_g'\}_{g \in G}$ acting on $\caC'_{x-1}$ as $ \tilde \beta_g \otimes \Ad{L(g)}$ for each $x \in 11r \Z$, where $ \tilde \beta = \pi_x \circ \beta|_{\caP_x^{[r]}} \circ \pi^{-1}_x$ and $g \mapsto L(g)$ is the left regular representation of $G$ on $\C[G]$. {We choose the corresponding genuine on-site representations and denote their family by 
$\sfU'$, so $\beta'=\beta^{\sfU'}$.
}

\subsubsection{Symmetric blending}

Consider the symmetric entangler $(\alpha \otimes \id, \sfU \otimes \sfU', \AA \otimes \caC')$. By definition, it has associated boundary algebras
$$\caP_x^{[r]} \otimes \caC'_{\partial_r\Left{x}}.$$

Moreover, since $\AA_{\partial_r \Left x} $ and $\caP_x^{'[r]}$ are $G$-invariant (with respect to $\beta$ and $\beta'$ respectively) full matrix algebras with the same dimension and the same (trivial) index, Lemma \ref{lem.isomorphism-projective-reps} provides a $G$-equivariant isomorphism 
\begin{equation} \label{eq.iota-1d}
\iota_x: \caP^{'[r]}_x \otimes \End \mathbb C[G] \to \AA_{\partial_r \Left x} \otimes \End \mathbb C[G].
\end{equation}

\begin{lemma} [Symmetric blending in 1d] \label{lem.symmetric-blending-1d}
    For each $x \in 11r\ZZ$, there is a symmetric blending between $(\alpha \otimes \id, \sfU \otimes \sfU', \AA \otimes \caC')$ and $(\id, \sfU \otimes \sfU', \AA \otimes \caC')$ at site $x$. That is, a symmetric entangler $(\theta_x,\sfU \otimes \sfU', \AA \otimes \caC')$ of range at most $2r$ such that 
    \begin{align} \label{eq.symmetric-blending-defining-prop}
        \theta_x|_{\caA_{\Left{x-2r}}} = \alpha|_{\caA_{\Left{x-2r}}} \otimes \id, \qquad \theta_x|_{(\caA\otimes \caC')_{\Left{x}^c}} = \id \otimes \id.
    \end{align}
\end{lemma}

\begin{proof}
    We use the notation introduced above, and particularly the definitions \eqref{eq.pi-1d} and \eqref{eq.iota-1d}. Let 
    $$
    \theta_x := ( \theta^{[2]}\circ\theta^{[1]}) \otimes \id_{\Left{x}^c},
    $$
    where $\theta^{[1]}$ and $\theta^{[2]}$ (depending on $x$) are defined as follows: 
    \begin{enumerate}
        \item The homomorphism $\theta^{[1]}: \AA_{\Left x} \otimes \caC'_{x-1} \to \al(\AA_{\Left x}) \otimes \caC'_{x-1}$ is defined by
    \[
    \theta^{[1]}:\quad
    \begin{tikzcd}[column sep=1.8em, row sep=4em]
       \AA_{\Left x} \arrow[d, "\alpha"'] 
      & \otimes \arrow[d, phantom]
      & \caP^{'[r]}_x \arrow[d, "\mathrm{id}"']
      & \otimes \arrow[d, phantom]
      &\End \mathbb C[G]\arrow[d, "\mathrm{id}"'] 
       \\
       \alpha(\AA_{\Left x})
      & \otimes 
         & \caP^{'[r]}_x
      & \otimes 
      &\End \mathbb C[G].
    \end{tikzcd}
    \]
    
        \item Given the identity
    $
     \alpha(\caA_{\Left x}) =  \caP_x^{[r]} \otimes \caA_{\Left{x-r}}
    $, we define $\theta^{[2]}: \al(\AA_{\Left x}) \otimes  \caC'_{x-1} \to \AA_{\Left x} \otimes \caC'_{x-1}$ by:
    \[
    \theta^{[2]}:\quad
    \begin{tikzcd}[column sep=1.1em, row sep=4em]
    \AA_{ \Left{x-r}} \arrow[d, "\id"'] 
      & \otimes \arrow[d, phantom] 
      & \mathcal P_x^{[r]} 
     \arrow[d, "\pi_x"'] 
      & \otimes \arrow[d, phantom] 
      & \caP^{'[r]}_x \otimes \End \mathbb C[G] \arrow[d, "\iota_x"'] \\
      \AA_{\Left{x-r}} 
      & \otimes 
      & \caP^{'[r]}_x
      & \otimes 
      & \mathcal A_{\partial_r \Left x} \otimes \End \mathbb C[G].
    \end{tikzcd}
    \]
    \end{enumerate}
    Since all maps are manifestly $G$-equivariant, so is $\theta_x$. For $A \in \caA_{L_{x-2r}}$ we have $\theta^{[1]}(A) = \al(A) \in \caA_{L_{x-r}}$ and therefore $\theta_x(A) = (\theta^{[2]} \circ \theta^{[1]})(A) = \al(A)$. Since $\theta_x|_{(\caA\otimes \caC')_{L_x^c}} = \id$ by definition, this proves Eq. \eqref{eq.symmetric-blending-defining-prop}. Since $\theta_x$ blends to the identity, its GNVW index is trivial so this also implies (\cite[~Theorem 9]{Gross_2012}) that $\theta_x$ is a FDQC with range at most $2r$.
\end{proof}

\subsection{Injectivity of the index map in 1d} \label{sec.disentangling-symmetric-entanglers-1d}

\begin{lemma}\label{lem:injectiveness-1d}
    The index map $\ind: \SE_{G, 1} \to H^2(G,U(1))$ is injective.
\end{lemma}

\begin{proof}
    The crux of the proof is to show that 1d symmetric entanglers with trivial index are stably equivalent to $\id$. Let us prove that first.
    
    Let $(\alpha,\sfU,\AA)$ be a symmetric entangler in 1d such that $\ind{(\alpha)} = [1] \in H^2(G,U(1))$, as in the previous section. 
    Lemma \ref{lem.symmetric-blending-1d} gives us symmetric entanglers $(\theta_x, \sfU \otimes  \sfU', \AA \otimes \caC')$ satisfying \eqref{eq.symmetric-blending-defining-prop}.

    We define symmetric entanglers
    \begin{equation*}
    \left(\al_k, \sfU \otimes \sfU', \AA \otimes \caC'  \right), \qquad k \in \ZZ,
    \end{equation*}
    by $\al_k := \theta_{11rk}^{-1} \circ \theta_{11r(k+1)}$, have range at most $4r$ and satisfy (see Figure \ref{fig:alpha-k-1d})
    \begin{align} 
        \nonumber \alpha_k|_{\caA_{\Left{11rk-2r}}} = \id, \quad &\alpha_k|_{\caA_{\Left{11r(k+1)}^c}} = \id, \\ \label{eq.alpha-k-properties in 1d}
        \alpha_k|_{\AA_{ [ 11rk+r, 11r(k+1) - 2r - 1] }} &= \alpha|_{\AA_{ [ 11rk+r, 11r(k+1) - 2r -1] }}.
    \end{align}
    The $\al_k$ for $k \in 2\Z$ therefore have disjoint supports so
    $$
    \bigg(\alpha_{\text{even}}= \prod_{k \in 2\ZZ} \alpha_k, \, \sfU \otimes \sfU', \, \AA \otimes \caC' \bigg)
    $$
    is a symmetric block partitioned FDQC with $\range \al_{\rm{even}} \leq 4r$. Now note that for $A \in \caA_{[11rk + 2r, 11r(k+1) - 3r -1]}$ we have $\al^{-1}( A ) \in \caA_{[ 11rk+r, 11r(k+1) - 2r -1]}$ and therefore $\al_k^{-1}(A) = \al^{-1}(A)$. It follows that the symmetric FDQC
    $$ \bigg(\alpha_{\text{odd}} = (\alpha \otimes \id) \circ \alpha_{\text{even}}^{-1}, \, \sfU \otimes \sfU', \, \AA \otimes \caC' \bigg) $$
    acts as $\id$ on the intervals $[11rk + 2r, 11r(k+1) - 3r -1]$ for $k \in 2\Z$ (see Figure \ref{fig:odd-even-1d}). Since these intervals have length $6r -1$ and $\range \al_{\rm{odd}} \leq 5r$ we conclude that $\al_{\rm{odd}}$ is also a symmetric block partitioned FDQC.
    
    We now have \begin{equation}\label{eq:stripes 1d}
            \alpha \otimes \id = \alpha_{\text{odd}} \circ \alpha_{\text{ even}}
    \end{equation}
    where each factor $\alpha_{\text{even}}$, $\alpha_{\text{odd}}$ is stably equivalent to $\id$ by Lemma \ref{lem.equivalence-block-partitioned-1d}. This means that $\al_{\rm{even}} \otimes \id$ and $\al_{\rm{odd}} \otimes \id$ can be cast as FDQCs with symmetric gates, after stacking with some auxiliary spin chain. Then clearly the composition $\al_{\rm{odd}} \circ \al_{\rm{even}} \sim (\al_{\rm{odd}} \circ \al_{\rm{even}}) \otimes \id = (\al_{\rm{odd}} \otimes \id) \circ (\al_{\rm{even}} \otimes \id) \sim \id$ is also stably equivalent to a FDQC with symmetric gates, and therefore stably equivalent to $\id$, as required.

    To finish the proof, suppose $\al$ and $\al'$ are two symmetric entanglers with $[\mu] = \ind(\al) = \ind(\al')$. By Lemma \ref{lem.surjectivity-1d} there is a symmetric entangler $\delta$ with $\ind(\delta) = -[\mu]$. By additivity of the index (Lemma \ref{lem.index-1d-symmetric-entanglers}) we have that $\delta \otimes \al'$ has trivial index and is therefore stably equivalent to $\id$. It follows that $\al \sim \al \otimes \id \sim \al \otimes \delta \otimes \al'$. But also $\al \otimes \delta$ has trivial index, and is therefore stably equivalent to $\id$. we therefore obtain $\al \sim \id \otimes \al' \sim \al'$, as required. 
\end{proof}

\begin{figure}
    \centering
    \begin{tikzpicture}[x=1cm, y=1cm, scale=1.2, transform shape]
        
        
        \draw[very thick, dotted] (-6,0) -- (7,0);
        \draw[very thick, red]  (-5,0) -- (-2.5,0);
        \draw[very thick, green] (-2.5,0) -- (3,0);
        \draw[very thick, red]  (3,0) -- (5,0);
        
        \tiny
        \node[below] at (-5,0) {$(11k-2)r$};
        \node[below] at (-2.5,0) {$(11k+1)r$};
        \node[below] at (3,0) {$(11k+9)r$};
        \node[below] at (5,0) {$(11k + 11)r$};
        
    \end{tikzpicture}
    \caption{The automorphism $\alpha_k$ acts by conjugation with a unitary block. The colored lines represent a non-trivial action. In the bulk (green region), $\alpha_k$ acts as $\alpha$.}
    \label{fig:alpha-k-1d}
\end{figure}
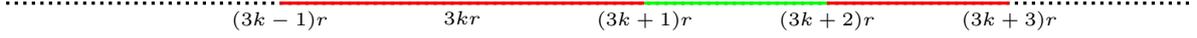

\vspace{.7cm}
\begin{figure}
    \centering
    \begin{tikzpicture}[x=1cm, y=1cm, scale=0.43]
    
    \begin{scope}[yshift=0cm]
    
    \draw[very thick, dotted] (-6,0) -- (28,0);
    
    \draw[very thick, red]  (-3,0) -- (2,0);
    \draw[very thick, green] (2,0) -- (8,0);
    \draw[very thick, red]  (8,0) -- (12,0);

    \draw[very thick, red]  (18,0) -- (23,0);
    \draw[very thick, green] (23,0) -- (28,0);
    
    \node[left] at (-6,0) {$\al_{\rm{odd}}$};
    
    %
    
    
    \end{scope}
    
    \begin{scope}[yshift=-1.5cm]
    
    
    \draw[very thick, dotted] (-6,0) -- (28,0);

    \draw[very thick, green] (-6,0) -- (-2,0);
    \draw[very thick, red]  (-2,0) -- (1,0);
    
    \draw[very thick, red]  (9,0) -- (11,0);
    \draw[very thick, green] (11,0) -- (19,0);
    \draw[very thick, red]  (19,0) -- (22,0);

    \node[left] at (-6,0) {$\al_{\rm{even}}$};
    \end{scope}
    \end{tikzpicture}
    
    \caption{The block partitioned automorphisms $\alpha_{\text{even}}$ and $\alpha_{\text{odd}}$, each illustrated on a spin chain. Non-trivial actions are represented by colored lines. Analogously as in Figure \ref{fig:alpha-k-1d}, they act as $\alpha$ on green regions.}
    \label{fig:odd-even-1d}
\end{figure}

\section{Classification of 2d symmetric entanglers} \label{sec: 2d classification}

\subsection{An index for 2d symmetric entanglers} \label{sec.index}


Throughout this section we consider a symmetric entangler $(\al, \sfU, \AA)$ on a 2-dimensional spin system $\caA = \otimes_{j \in \Z^2} \caA_j$.

\subsubsection{The boundary algebra}

We construct boundary algebras associated to $(\alpha, \sfU, \AA)$ which generalize the boundary algebras constructed in Section \ref{sec.boundary-algebra-1d} for the one-dimensional case. \\

{Fix a circuit representation $\mathsf{A}$ for $\alpha$, of depth $n$ and maximum layer range $q$, and choose an integer $r\geq\max\{nq,1\}$. We retain this choice of circuit and scale in the subsequent boundary-algebra and blending constructions.}  Then $\range(\al) \leq \range(\mathsf{A}) \leq r$.

For $x \in \Z$ we let $\Left {x}$ be the half-plane
$$
\Left x=\{ (i_1,i_2)\ |\ i_1 < x \} \subset \Z^2,
$$
and we define the boundary region
\begin{equation} \label{eq.boundary-region}
    \partial_r \Left x := \{ (i_1,i_2)\ |\ x-r \le i_1 < x\}.
\end{equation}
The \emph{boundary algebra} is defined as
\begin{equation} \label{eq.boundary-algebra-hastings-freedman}
   \caP_x^{[r]} :=  \alpha (\AA_{\Left x}) \cap \AA_{ \Left{x-r}}'.
\end{equation}

Define $(\eta_x,\AA)$ to be the FDQC associated to the circuit datum $\mathsf{B}$ defined as follows:
\begin{itemize}
    \item From the first layer of $\mathsf{A}$, we retain gates which connect the half-plane with its complement, i.e. whose support intersects both $\Left x$ and $\Left{x}^c$. 
 These gates form the first layer of $\mathsf{B}$. Let $B_1$ be the union of the supports of all those gates. 

    \item For $i=2,\ldots,n$, the $i^{\mathrm{th}}$ layer of $\mathsf{B}$ consists of all gates in the $i^{\rm{th}}$ layer of $\mathsf{A}$ whose support intersects $B_{i-1}$ or which connect the half-plane to its complement. We call the union of $B_{i-1}$ with the supports of these new gates $B_{i}$. 
\end{itemize}
{The same gate-selection procedure is illustrated in one dimension in Figure~\ref{fig:eta1d}.}
With this definition of $\eta_x$ we get
\begin{equation} \label{eq.eta-boundary-algebra-2d}
\alpha(\caA_{\Left x})=\eta_x(\caA_{\Left x})= \eta_x(\caA_{\partial_r \Left x}) \otimes \caA_{\Left {x-r}} = \caP_x^{[r]}  \otimes \caA_{\Left {x-r}}.
\end{equation}
By analogous reasoning as in Lemma \ref{lem.invariant-boundary}, the boundary algebras $\caP_x^{[r]}$ are $\beta_g$-invariant.

For future use, we note that
\begin{equation}\label{eq: split alpha l}
\alpha|_{\caA_{\Left x}}= \eta_x|_{\caA_{\Left x}} \circ \alpha'_{[x]} 
\end{equation}
where $\alpha'_{[x]}$ is the FDQC with circuit representative consisting of those gates of $\mathsf{A}$ whose support {is contained in} $\Left x$ and that were not chosen in the circuit representation $\sfB$ defining $\eta_x$ (the red gates in Figure \ref{fig:eta1d}).

\subsubsection{A locality preserving symmetry on the boundary algebra} \label{subsubsec:LPS on P}

We will consider spin chains oriented along the boundary $\partial_r \Left x$. To that end, note that $\partial_r \Left x$ as defined in \eqref{eq.boundary-region} is a strip. For the sake of this subsection, it will be advantageous to view this strip as having unit ``width'', and so we put, for $\J \in \bbZ$:
$$
\caA_{x,\J}^{[r]}=\caA_{\{(x-r,\J),\ldots,(x-1,\J)\}}
$$
and define $\caP_{x,\J}^{[r]}=\eta_x(\caA_{x,\J}^{[r]})$ (the algebras $\caP_{x,\J}^{[r]}$ still depend on $\alpha$, which we omit for clarity of notation). Then
$$\caP_x^{[r]}=\otimes_\J \caP_{x,\J}^{[r]}$$
is a spin chain. Even though $\beta_g$ was an on-site action on $\caA$, it is not necessarily on-site when restricted to the spin chain $\caP_x^{[r]}$, since the notion of ``sites'' is different. Namely, in general, 
$$\beta_g (\mathcal P_{x,\J}^{[r]}) = \beta_g \circ \eta_x (\AA_{x,\J}^{[r]}) \not\subset \mathcal P_{x,\J}^{[r]},$$
since $\eta_x$ is not necessarily equivariant. However, each $\beta_g$ acts as a QCA on $\mathcal P_{x}^{[r]}$, from which we obtain the following lemma: 

\begin{lemma} \label{lem.LPS-boundary-algebra}
    The family $\{\beta_g|_{\caP_x^{[r]}}\}_{g \in G}$ is an LPS on the spin chain $\caP_{x}^{[r]} = \otimes_\J \caP_{x,\J}^{[r]}$. The range of $(\beta, \caP_x^{[r]})$ is bounded above by {$2r$}.
\end{lemma}
\begin{proof}
{The invariance of $\caP_x^{[r]}$ was established above. Under the sitewise identification $\eta_x:\caA_{\partial_r L_x}\to\caP_x^{[r]}$, its symmetry action is $\eta_x^{-1}\circ\beta_g\circ\eta_x$. Since $\eta_x$ and $\eta_x^{-1}$ have range at most $r$, and $\beta_g$ is on-site on $\caA$, this action has range at most $2r$.}
\end{proof}

Recall that we write $\Sym_{G, 1}$ for the set of locality preserving symmetries on spin chains. In \cite{bols2025classificationlocalitypreservingsymmetries} an index $\Omega : \Sym_{G, 1} \to H^3(G, U(1))$ was defined which completely classifies LPSs on spin chains up to stable equivalence. We recall the index and revise its properties in Appendix \ref{appendix.LPS}, see e.g. Proposition \ref{thm.classification-1d-LPS} on the classification. We will now show that $\Omega(  (\beta, \caP_x^{[r]}) )$ is a well-defined index for our symmetric entangler $\al$. We will often simply write $\beta$ instead of $(\beta, \caP_x^{[r]})$ when it is clear from context that $\beta$ is to be interpreted as an LPS on a boundary algebra.

\subsubsection{An index for 2d symmetric entanglers}

We define the index 
\begin{equation}
\ind(\alpha) := \Omega(\beta) \in H^{3}(G,U(1))
\end{equation}
of the symmetric entangler $(\alpha, \sfU, \AA)$ as the index of the corresponding locality preserving symmetry $(\beta,\caP_{x}^{[r]})$ from Lemma \ref{lem.LPS-boundary-algebra}.

\begin{lemma} \label{lem.index-2d-does-not-depend-FDQC}
    The index is well-defined. That is, the class $\Omega(\beta)$
    \begin{enumerate}
        \item does not depend on the choice of $r \geq nq$, for fixed circuit representation of $\al$;
        \item does not depend on the choice of circuit representation for $(\alpha,\sfU,\AA)$;
        \item does not depend on $x \in \ZZ$.
    \end{enumerate}
\end{lemma}

\begin{proof}
    Items 1 and 3 follow from analogous arguments as in the proof of Lemma \ref{lem.index-independent-of-r-1d}.

    For item 2, let $\tilde \eta_x$ be obtained from another circuit representation of $\alpha$ of depth $\tilde n$ and maximal range of its layers $\tilde q$. Take $r \geq \max\{ nq, \tilde n \tilde q, 1\}$. We have that
    \begin{align*}
        \mathcal P_x^{[r]} = \otimes_\J \caP_{x,\J}^{[r]} = \otimes_{\J} \tilde \caP_{x,\J}^{[r]}
    \end{align*}
    are two spin chain representations of $\caP_x^{[r]}$ in terms of on-site algebras $\caP_{x,\J}^{[r]} = \eta_x(\AA_{x,\J}^{[r]})$ and $\tilde \caP_{x,\J}^{[r]} = \tilde \eta_x(\AA_{x,\J}^{[r]})$. We consider the LPS $({\mathscr F^{-1} \circ \beta \circ \mathscr F}, \otimes_\J \caP_{x,\J}^{[r]})$, where $\mathscr F = \tilde \eta_x \circ \eta_x^{-1}$ is a spin system isomorphism. Since $\mathscr F (\caP_{x,\J}^{[r]}) = \tilde \caP_{x,\J}^{[r]}$, and $\mathscr F$ intertwines the $G$-actions $\beta$ and $\mathscr F^{-1} \circ \beta \circ \mathscr F$, we have that
    $(\beta, \otimes_{\J} \tilde \caP_{x,\J}^{[r]})$ is equivalent to $ ({\mathscr F^{-1} \circ \beta \circ \mathscr F}, \otimes_\J \caP_{x,\J}^{[r]}).$
    
    Notice that $\mathscr F \in \QCA(\caP_{x}^{[r]})$. By Lemma \ref{lem.equivalence-symmetries-qca} (Stable equivalence under conjugation by a QCA), we have that 
    $$(\beta, \otimes_{\J} \tilde \caP_{x,\J}^{[r]}) \sim ({\mathscr F^{-1} \circ \beta \circ \mathscr F}, \otimes_\J \caP_{x,\J}^{[r]}) \sim (\beta,\otimes_\J \caP_{x,\J}^{[r]}),$$
    so item 2 follows from the invariance of $\Omega$ under stable equivalence of LPS.
\end{proof}

\begin{proposition} \label{prop.index}
    The index $\ind$ descends to a homomorphism of monoids
    \begin{align*}
        \ind: \SE_{G, 2} \longrightarrow H^3(G,U(1)).
    \end{align*}
\end{proposition}
\begin{proof}
    Completely analogous to the proof of Lemma \ref{lem.index-1d-symmetric-entanglers}.
\end{proof}

\subsection{Surjectivity of the index map in 2d}

\begin{lemma} [Surjectivity of the index map in 2d] \label{lem.surjectivity-2d}
    For each class $[\omega] \in H^3(G,U(1))$, there exists a 2d symmetric entangler $(\alpha, \sfU,\AA)$ such that $\ind(\alpha) = [\omega]$.  
\end{lemma}

\begin{proof}
    The following construction was suggested in \cite{seifnashri2025disentangling}. Fix a class $[\omega] \in H^3(G, U(1))$. From \cite[Section~4]{bols2025classificationlocalitypreservingsymmetries} we get a locality preserving symmetry $(\beta, \caA)$ on a spin chain $\caA = \otimes_{i \in \Z} \caA_i$ with index $\Omega(\beta) = [\omega]$, as well as a locality preserving symmetry $(\bar \beta, \caA)$ on the same spin chain with index $\Omega(\bar \beta) = -[\omega]$.
    
    {By Lemma \ref{lem.trivializing-LPS}, there are on-site LPSs
    $(\delta,\caD)$ and $(\delta',\caD')$ induced by on-site unitary
    representation data, and an FDQC
    $(\alpha_1,\widetilde\caA'\otimes\widetilde\caA)$, where
    \[
        \widetilde\caA=\caA\otimes\caD,
        \qquad
        \widetilde\caA'=\caA\otimes\caD',
    \]
    such that
    \[
        \kappa
        :=\alpha_1^{-1}\circ
        \bigl((\bar\beta\otimes\delta')
              \otimes(\beta\otimes\delta)\bigr)
        \circ\alpha_1
    \]
    is induced by an on-site unitary representation datum on the
    paired chain $\widetilde\caA'\otimes\widetilde\caA$.
    
    Having unitary representations on $\widetilde\caA'_x\otimes\widetilde\caA_x$, Lemma \ref{lem.isomorphism-projective-reps} allows us to conjugate this to a tensor product of unitary representation, after stacking with regular representations.  We can hence extend $(\delta,\caD), (\delta',\caD')$ by on-site regular representations, such that, after modifying  $\alpha_1$ accordingly,  we obtain the stronger conclusion 
    \[
        \gamma'\otimes\gamma
        :=\alpha_1^{-1}\circ
        \bigl((\bar\beta\otimes\delta')
              \otimes(\beta\otimes\delta)\bigr)
        \circ\alpha_1
    \]
    where both $\gamma', \gamma$ are induced by on-site unitary representation data on $\widetilde\caA', \widetilde\caA$, respectively. 
    
    Analogously, we find  an FDQC $(\al_2, \widetilde \caA \otimes \widetilde \caA')$ such that
    $$
    \gamma \otimes \gamma' = \al_2 \circ ( (\beta \otimes \delta) \otimes (\bar \beta \otimes \delta') ) \circ \al_2^{-1}.
    $$ }

    Consider the two-dimensional spin system $\caB = \bigotimes_{j \in \Z} \caB_j$ where the rows parallel to the $y$-axis are given by
    $$ \caB_j \simeq \begin{cases} \widetilde \caA & \text{if } j \, \text{even} \\ \widetilde \caA' & \text{if } j \, \text{odd}, \end{cases}$$
    see Figure \ref{fig:exampleSPT2d} (which is however rotated $90^{\circ}$ w.r.t. the usual presentation of the $xy$-plane). For each $j \in \Z$ we let $(\gamma_{2j}, \caB_{2j})$ be an isomorphic copy of $(\gamma, \widetilde \caA)$, and we let $(\gamma_{2j+1}, \caB_{2j+1})$ be an isomorphic copy of $(\gamma', \widetilde \caA')$. These combine into an on-site symmetry $(\Gamma, \caB)$ with $\Gamma = \otimes_{j \in \Z} \gamma_j$.

    Similarly, for each $j \in \Z$ let the LPS $(\beta_{2j}, \caB_{2j})$ be an isomorphic copy of $(\beta \otimes \delta, \widetilde \caA)$, and let $(\beta_{2j+1}, \caB_{2j+1})$ be an isomorphic copy of $(\bar \beta \otimes \delta', \widetilde \caA')$. These combine into a locality preserving symmetry $B = \otimes_{j \in \Z} \beta_j$ on $\caB$.
    
    Finally, for each $j \in \Z$ let $(\al_1^{(2j-1, 2j)}, \caB_{2j-1} \otimes \caB_{2j})$ be an isomorphic copy of $(\al_1, \widetilde \caA' \otimes \widetilde \caA)$, and let $(\al_2^{(2j, 2j+1)}, \caB_{2j} \otimes \caB_{2j+1})$ be an isomorphic copy of $(\al_2, \widetilde \caA \otimes \widetilde \caA')$. These combine into FDQCs $A_1 = \otimes_{j \in \Z} \al_1^{(2j-1, 2j)}$ and $A_2 = \otimes_{j \in \Z} \, \al_2^{(2j, 2j+1)}$ with the following properties:
    $$ B = A_1 \circ \Gamma \circ A_1^{-1}, \quad \Gamma =  A_2 \circ B \circ A_2^{-1}.$$

    We consider now the FDQC $( A := A_2 \circ A_1, \caB )$. This FDQC is equivariant under the on-site symmetry $\Gamma$, indeed
    $$ A \circ \Gamma = A_2 \circ A_1 \circ \Gamma =  A_2 \circ B \circ A_1 = \Gamma \circ A_2 \circ A_1 = \Gamma \circ A.$$
    
    We now compute its index using the left  half-plane $L = \Left{1}$. Writing $\caB_{\leq j} := \caB_{\Left{j+1}}$, 
    $$ A(\caB_{L}) = A( \caB_{\leq 0} ) = A_2 \big( A_1(\caB_{\leq 0}) \big) = A_2( \caB_{\leq 0} ) = \caB_{\leq -1} \otimes \al_2^{(0, 1)}( \caB_0 ). $$
    From this it is clear that we can take $\caP = \al_2^{(0, 1)}( \caB_0 )$ as our boundary algebra. To see how the global symmetry $\Gamma$ acts on this boundary algebra, take ${a} \in \caB_0$. Then
    $$ \Gamma \big(  \al_2^{(0, 1)}({a})  \big) = \al_2^{(0, 1)} \big( \beta_0({a}) \big). $$
    We see that $( \Gamma|_{\caP}, \caP)$ is isomorphic as an LPS to $(\beta \otimes \delta, \widetilde \caA)$, which has index $[\omega]$. We conclude that the equivariant FDQC $(A, \caB)$ has anomaly $[\omega]$.
    \begin{figure}
        \centering
    \begin{tikzpicture}[>=stealth]
    
        \def\lines{6}
        \def\dotsperline{4}
        \def\spacing{0.6} 
        \def\groupspacing{1.5} 

        \pgfmathsetmacro{\yoffset}{-(-1) * (\lines * \spacing )}
        
        \foreach \l in {1,...,\lines} {
    
            \ifodd\l
                \def\linecolor{blue!60}
                \def\labell{$\beta$}
            \else
                \def\linecolor{green!60}
                \def\labell{$\bar\beta$}
            \fi 
            
            \pgfmathsetmacro{\currenty}{\yoffset - (\l-1) * \spacing}
            
            \draw[\linecolor, thin] (0, \currenty) -- (\dotsperline , \currenty);
    
            \node[] at (\dotsperline / 2,\currenty ) {\labell};
            
            \foreach \x in {0,...,\numexpr\dotsperline-1\relax} {
                \fill[\linecolor] ( \x + 0.5, \currenty) circle (2pt);
            }
        }
    
    
        \foreach \l in {1,2,3}{
        \pgfmathsetmacro{\y}{ \l * 2*  \spacing }
            \draw [decorate, decoration={brace, amplitude=5pt, raise=5pt}, thick]
                ( \dotsperline , \y ) -- (\dotsperline , \y - 1 * \spacing);
    
                \draw [decorate, decoration={brace, amplitude=5pt, raise=5pt}, thick]
                ( \dotsperline +1.5 , \y - 1 * \spacing ) -- (\dotsperline +1.5, \y );
        
            \draw[->] (\dotsperline +0.5 , \y - 0.5 * \spacing) -- (\dotsperline +1 , \y - 0.5 * \spacing) 
                node[midway, above] {$\alpha_2$};
        }

    
        \foreach \l in {1,2}{
        \pgfmathsetmacro{\y}{ \l * 2*  \spacing + \spacing}
            \draw [decorate, decoration={brace, amplitude=5pt, raise=5pt}, thick]
                ( -\groupspacing  , \y ) -- (-\groupspacing  , \y - 1 * \spacing);
    
                \draw [decorate, decoration={brace, amplitude=5pt, raise=5pt}, thick]
                ( -\groupspacing  +1.5 , \y - 1 * \spacing ) -- (-\groupspacing  +1.5, \y );
        
            \draw[->] (-\groupspacing  +0.5 , \y - 0.5 * \spacing) -- (-\groupspacing  +1 , \y - 0.5 * \spacing) 
                node[midway, above] {$\alpha_1$};
        }
    
        
        \pgfmathsetmacro{\yoffset}{-(-1) * (\lines * \spacing )}
        
        \foreach \l in {1,...,\lines} {
    
            \ifodd\l
                \def\linecolor{orange!60}
                \def\labell{$\gamma$}
            \else
                \def\linecolor{red!60}
                \def\labell{$\gamma'$}
            \fi 
            
            \pgfmathsetmacro{\currenty}{\yoffset - (\l-1) * \spacing}
            
            \draw[\linecolor, thin] (\dotsperline + \groupspacing , \currenty) -- (\dotsperline + \groupspacing + \dotsperline , \currenty);
    
            \node[] at (\dotsperline + \groupspacing +\dotsperline / 2,\currenty ) {\labell};
            
            \foreach \x in {0,...,\numexpr\dotsperline-1\relax} {
                \fill[\linecolor] ( \dotsperline + \groupspacing + \x + 0.5, \currenty) circle (2pt);
            }
        }
    
        
        \pgfmathsetmacro{\yoffset}{-(-1) * (\lines * \spacing )}
        
        \foreach \l in {1,...,\lines} {
    
            \ifodd\l
                \def\linecolor{orange!60}
                \def\labell{$\gamma$}
            \else
                \def\linecolor{red!60}
                \def\labell{$\gamma'$}
            \fi 
            
            \pgfmathsetmacro{\currenty}{\yoffset - (\l-1) * \spacing}
            
            \draw[\linecolor, thin] (-\dotsperline - \groupspacing , \currenty) -- (-\dotsperline - \groupspacing + \dotsperline , \currenty);
    
            \node[] at (-\dotsperline - \groupspacing +\dotsperline / 2,\currenty ) {\labell};
            
            \foreach \x in {0,...,\numexpr\dotsperline-1\relax} {
                \fill[\linecolor] ( -\dotsperline - \groupspacing + \x + 0.5, \currenty) circle (2pt);
            }
        }
     \def\y{-0.2}
     \node[] at (\dotsperline/2, \y) {$B$};
     \node[] at (3*\dotsperline/2 + \groupspacing, \y) {$\Gamma$};
     \node[] at (-\groupspacing - \dotsperline /2, \y) {$\Gamma$};
    \draw[->, thick] (-\groupspacing  +0.5 , \y ) -- (-\groupspacing  +1 , \y ) 
                node[midway, above] {$A_1$};
    \draw[->, thick] (\dotsperline  +0.5 , \y ) -- (\dotsperline  +1 , \y ) 
                node[midway, above] {$A_2$};
    
    \end{tikzpicture}
        \caption{An illustration of the build-up of the example class. In order to better represent the actions of $A_1$ and $A_2$, we take the $y$-axis to point to the right, and the $x$-axis pointing down. Each square lattice represents the spin system $\caB$, with different on-site symmetries, given by respectively $\Gamma, B$ and $\Gamma$. The entanglers between the symmetries are $A_1$ and $A_2$, acting componentwise on the indicated stripes.}
        \label{fig:exampleSPT2d}
    \end{figure}
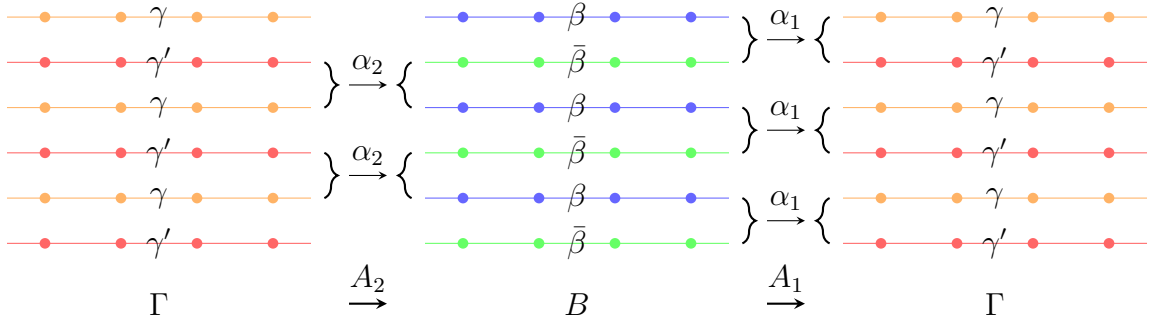
\end{proof}

\subsection{Symmetric blending in 2d} \label{sec.symmetric-blending}

The goal of this section is to construct symmetric blendings $\theta_x$ between a symmetric entangler with trivial $H^3(G,U(1))$ index, and the identity automorphism on $\AA$.

{Assume that $\ind(\alpha)=[1]$, and retain the circuit representation $\mathsf A$ and the integer $r\geq\max\{nq,1\}$ chosen in Section \ref{sec.index}. 
}

\subsubsection{Disentangling $\beta$ on the spin chain $\caP_x^{[r]}$}

Recall that, by Lemma \ref{lem.LPS-boundary-algebra}, the boundary algebra $\caP_x^{[r]} =\otimes_\J \caP_{x,\J}^{[r]}$ is a spin chain that carries an action of $G$ by the LPS $(\beta, \caP_{x}^{[r]})$. By Lemma \ref{lem.trivializing-LPS}, there is an on-site LPS $(\delta_x,\mathcal D_x)$ {induced by an {on-site unitary representation datum}} on a spin chain $\caD_x = \otimes_\J \caD_{x,\J}$, and an FDQC $(\gamma_x, \caP_x^{[r]} \otimes \mathcal D_x)$ with range of $\mathcal{O}(r)$ so that 
$$\gamma_x \circ (\beta \otimes \delta_x) \circ \gamma_x^{-1}$$
{is induced by an {on-site unitary representation datum}} on $\caP_x^{[r]} \otimes \caD_x$. Moreover, the on-site dimensions $\dim \caD_{x, y}$ are uniformly bounded in terms of $r$, $G$, and $\sup_{x, y} \dim \caP_{x, y}^{[r]}$.
Let us define the algebras
$$
\caC_{x,\J} = \gamma_x^{-1}(\caP_{x,\J}^{[r]} \otimes \caD_{x,\J})
$$
and note that $(\beta \otimes \delta_x)$ is on-site with respect to the decomposition $\caC_x := \gamma_x^{-1}(\caP_x^{[r]} \otimes \caD_x) = \otimes_\J \, \caC_{x,\J}$. That is, $(\beta \otimes \delta_x) (\caC_{x,\J})=\caC_{x,\J}$. We introduce isomorphic copies of the $\caC_{x,\J}$ which we call $\caC'_{x,\J}$, and define the spin chain
$$
\caC'_x=\otimes_\J \, \caC_{x,\J}'.
$$

We fix a spin system isomorphism $\pi_x : \caC_x \rightarrow \caC'_x$ 
that takes each $\caC_{x,\J}$ to $\caC'_{x,\J}$ and equip $\caC'_x$ with the on-site symmetry $\beta'_x = \{ \beta'_{x, g} \}_{g \in G}$ given by
$$
    \beta'_x := \pi_x \circ (\beta \otimes \delta_x) \circ \pi_x^{-1}.
$$

Finally, we want to obtain an equivariant spin system isomorphism $\iota_x$ from $\caC'_x$ to $\otimes_{\J} \, (\caA^{[r]}_{x,\J} \otimes \caD_{x,\J})$ which maps each $\caC'_{x,\J}$ to $\caA^{[r]}_{x,\J} \otimes \caD_{x,\J}$. 
It is always possible to find such an isomorphism inside a stable equivalence class. Indeed, the finite-dimensional full matrix algebras $\mathcal C'_{x,\J}$ and $\mathcal A^{[r]}_{x,\J} \otimes \mathcal D_{x,\J}$ are $G$-invariant under $\beta'_x$ and $\beta \otimes \delta_x$ respectively, and these symmetry actions are both linear. Moreover, recalling that $\caP^{[r]}_{x,\J} = \eta_x(\caA^{[r]}_{x,\J})$, we see that $\mathcal C'_{x,\J}$ and $\mathcal A^{[r]}_{x,\J} \otimes \mathcal D_{x,\J}$ have the same dimension. 
We can therefore apply Lemma \ref{lem.isomorphism-projective-reps} to see that we can stack $\AA$ with a spin chain along the boundary $\partial_r \Left x$ carrying on-site left regular representations of $G$ and replace $\al$ by $\alpha \otimes \id$, so that the $G$-actions on $\mathcal C'_{x,\J}$ and $\mathcal A^{[r]}_{x,\J} \otimes \mathcal D_{x,\J}$ are isomorphic. We assume without loss of generality that these regular ancillary degrees of freedom are included in $\caA$, so we have $G$-equivariant isomorphisms $\iota_x|_{\caC'_{x,\J}} : \caC'_{x,\J} \to \caA^{[r]}_{x,\J} \otimes \caD_{x,\J}$ which together define a $G$-equivariant isomorphism $\iota_x:\caC'_x \to\otimes_{\J} \, (\caA^{[r]}_{x,\J} \otimes \caD_{x,\J})$ as required.

\subsubsection{The blending maps $\theta_x$} \label{sec: blending map}

We recall that $\caA=\otimes_{j\in {\bbZ^2}} \caA_j$ is a 2d spin system and $\caD_x=\otimes_{\J \in \bbZ} \caD_{x,\J}$ and $\caC'_x=\otimes_{\J \in \bbZ} \caC'_{x,\J}$ are spin chains equipped with on-site symmetries $\beta$, $\delta_x$ and $\beta'_x$ respectively.
We will put these together in a 2d spin system by stacking the spin chains $\caD_x$ and $\caC'_x$ along the boundaries of $\partial_r \Left x$, {for appropriate choices of $x$.}
More precisely, {for a sufficiently large natural $N>0$}, to be fixed later,  we consider the 2-dimensional spin system $\widetilde \caA$ with on-site degrees of freedom
\begin{align}
    \label{ex.auxiliary-algebra-blending-1} \widetilde \caA_{(x-1,\J)} &=\caA_{(x-1,\J)}\otimes \caD_{x,\J} \otimes  \caC'_{x,\J}, \qquad \forall x \in {Nr}\ZZ, \J \in \Z \\
    \label{ex.auxiliary-algebra-blending-2}   \widetilde \caA_{(x,y)} &=\caA_{(x,y)}, \qquad \qquad \qquad \qquad \quad \forall x \not\in {Nr}\ZZ-1, y \in \Z.
\end{align}
Since the on site dimensions of the chains $\caD_x$ and $\caC'_x$ are bounded uniformly in $x$, this is indeed a spin system.

{Let $\sfD$ denote a chosen family of genuine on-site representations implementing $\delta$, and let $\sfU'$ implement $\beta'$.} On $\widetilde \caA$ we denote the symmetry by $\tilde \beta := \beta \otimes \delta \otimes \beta'$ where $\delta = \otimes_x \delta_x$ acts on $\caD = \otimes_x \caD_x$ and $\beta' = \otimes_x\beta'_x$ acts on $\caC' = \otimes_x\caC'_x$. The following Lemma is a symmetric version of a particular case of \cite[Theorem 3.7]{freedman2020classification} (namely, the FDQC case).

\begin{lemma} [Symmetric blending in 2d]  \label{lem:symmetric_blending}
   For each $x \in {Nr\ZZ}$, there exists a symmetric entangler $(\theta_x, \sfU \otimes \sfD \otimes \sfU', \widetilde \caA)$ of range $\mathcal{O}(r)$ such that
    \begin{align} \label{eq.symmetric-blending-defining-prop-2d}
        \theta_x|_{\widetilde \caA_{\Left{x-2r}}} = \alpha|_{\caA_{\Left{x-2r}}} \otimes \id, \qquad \theta|_{\widetilde\AA_{\Left{x}^c}} = \id \otimes \id.
    \end{align}
    Moreover, $\theta_x$ has a circuit representation $\mathsf{\Theta}_x$ {of circuit range $\mathcal{O}(r)$, uniformly in $x$}, all of whose gates are supported in $L_x$ and for which $\mathsf{\Theta}_x|_{{L_{x-r}}} = \mathsf{A}|_{L_{x-r}}$.
\end{lemma}

\begin{proof}
The blending is defined similarly as in the 1d case (Lemma \ref{lem.symmetric-blending-1d}): We let 
$$
\theta_x := ( \theta^{[2]}\circ\theta^{[1]}) \otimes \id_{\caA_{\Left{x}^c}},
$$
where $\theta^{[1]}$ and $\theta^{[2]}$ are defined as follows: 
\begin{enumerate}
    \item Let $\theta^{[1]}: \AA_{\Left x} \otimes \mathcal D_x \otimes \mathcal C'_x \to \alpha(\AA_{\Left x}) \otimes \mathcal D_x \otimes \mathcal C'_x$ be the homomorphism acting according to the diagram:
\[
\theta^{[1]}:\quad
\begin{tikzcd}[column sep=1.8em, row sep=4em]
\AA_{\Left x} \arrow[d, "\alpha"'] 
  & \otimes \arrow[d, phantom] 
  & \mathcal{D}_x \arrow[d, "\mathrm{id}"'] 
  & \otimes \arrow[d, phantom] 
  & \caC'_x \arrow[d, "\mathrm{id}"'] \\
\alpha(\AA_{\Left x})  
  & \otimes 
  & \mathcal{D}_x 
  & \otimes 
  & \caC'_x
\end{tikzcd}
\]

    \item Recalling Eq. \eqref{eq.eta-boundary-algebra-2d}, the homomorphism $\theta^{[2]}: \al(\AA_{\Left x}) \otimes \mathcal D_x \otimes \mathcal C'_x \to \AA_{\Left x} \otimes \mathcal D_x \otimes \mathcal C'_x$ is defined by
\[
\theta^{[2]}:\quad
\begin{tikzcd}[column sep=1.1em, row sep=4em]
\AA_{\Left{x-r}} \arrow[d, "\id"'] 
  & \otimes \arrow[d, phantom] 
  &  \caC_x \arrow[d, "\pi_x"'] 
  & \otimes \arrow[d, phantom] 
  & \caC'_x \arrow[d, "\iota_x"'] \\
\AA_{\Left{x-r}}
  & \otimes 
  & \mathcal{\caC}_x' 
  & \otimes 
  & \mathcal A_{\partial_r L_x} \otimes \mathcal D_x.
\end{tikzcd}
\]

\end{enumerate}
All maps that enter are manifestly $G$-equivariant, and hence so is $\theta_x$. Note moreover that each $\caP_{x, y}^{[r]} = \eta_x(\caA^{[r]}_{x, \J})$ is a subalgebra of $\widetilde \caA_{\{(x, y)\}^{(r)}}$, and that $\gamma_x$ has range $\mathcal{O}(r)$. It follows that each $\caC_{x, y} = \gamma_x^{-1}(\caP_{x,\J}^{[r]} \otimes \caD_{x,\J})$ is a subalgebra of $\widetilde \caA_{\{(x, y)\}^{(Cr)}}$ for some $C$. Together with the fact that $\al$ has range bounded by $r$, this implies that $\theta_x$ has range $\mathcal{O}(r)$.

To show that $\theta_x$ is an FDQC, we first use \eqref{eq: split alpha l} to split $\theta^{[1]}=\gamma_x\circ\theta'\circ \alpha'$ where we write $\alpha' := \alpha'_{[x]}$, and
\[
\theta':\quad
\begin{tikzcd}[column sep=1.1em, row sep=4em]
\AA_{\Left {x-r}} \arrow[d, "\id"'] 
 & \otimes 
 &  \mathcal A_{\partial_r L_x} \otimes \mathcal D_x \arrow[d, "\gamma_x^{-1}\circ(\eta_x\otimes \id)"']
  & \otimes \arrow[d, phantom] 
  & \caC'_x \arrow[d, "\id"'] \\
 \AA_{\Left{x-r}} 
 &\otimes 
 & \mathcal{P}_x^{[r]} \otimes \mathcal D_x
  & \otimes 
  & \mathcal C'_x
\end{tikzcd}
\]

Note that $\alpha'$ is manifestly an FDQC.
The remaining map can be recast as
\begin{align*}
\theta_x \circ (\alpha')^{-1} &= \theta^{[2]} \circ\gamma_x\circ\theta' \\ &= (\pi_x \circ \gamma_x\circ\pi_x^{-1})\circ\theta^{[2]}\circ\theta'.
\end{align*}
Since $\gamma_x$ is an FDQC, so is $\pi_x\circ \gamma_x\circ\pi_x^{-1}$.
The crux is that $\theta^{[2]}\circ\theta'$ is on-site, and therefore also an FDQC.

To see this, note first that $\theta^{[2]} \circ \theta'$ acts as identity on $\caA_{L_{x-r}}$, so it is sufficient to check that it is  on-site on $\caA_{\partial_r L_x} \otimes \caD_x \otimes \caC'_x$. Indeed, for each $\J \in \ZZ$ we have
\begin{equation*}
    \mathcal{A}^{[r]}_{x,\J} \otimes \mathcal{D}_{x,\J} 
    \xrightarrow{\quad \eta_x \otimes \text{id} \quad} 
    \mathcal{P}_{x,\J}^{[r]} \otimes \mathcal{D}_{x,\J} 
    \xrightarrow{\quad \gamma_x^{-1} \quad} 
    \mathcal{C}_{x,\J} 
    \xrightarrow{\quad \pi_x \quad} 
    \mathcal{C}'_{x,\J}
\end{equation*}
and
\begin{equation*}
    \caC'_{x, \J} \, \xrightarrow{ \quad \iota_x \quad } \, \caA^{[r]}_{x, \J} \otimes \caD_{x, \J}.
\end{equation*}
Together, these facts imply that $(\theta^{[2]} \circ \theta')( \caA^{[r]}_{x, \J} \otimes \caD_{x, \J} \otimes \caC'_{x, \J}) = \caA^{[r]}_{x, \J} \otimes \caD_{x, \J} \otimes \caC'_{x, \J}$ for all $\J \in \Z$, as required.
We conclude that $\theta_x$ is an FDQC, and it has a circuit representation that agrees with $\alpha$ left of the cut. 

Moreover, we have $\theta_x = (( \pi\circ\gamma_x\circ\pi^{-1}) \circ \theta^{[2]} \circ\theta') \circ \al'$, and we have just shown that $\theta^{[2]} \circ \theta'$ is an FDQC with circuit representative $\mathsf{C_x}$ all of whose gates are on-site and belong to $\caA_{\partial_r L_x} \otimes \caD_x \otimes \caC'_x$, and which are therefore supported outside of $L_{x - r}$. By definition, $\al'$ has circuit representative $\mathsf{A}'$ which satisfies $\mathsf{A}'|_{L_{x-r}} = \mathsf{A}|_{L_{x-r}}$. 
{Next, we choose a circuit representation of $\gamma_x$.
By \cite{Gross_2012}, a one-dimensional FDQC of automorphism
range $R$ admits a depth-two circuit representation with
gate diameters $\mathcal O(R)$.
Since $\gamma_x$ has automorphism range $\mathcal O(r)$,
uniformly in $x$, we can therefore choose such a representation
$\mathsf\Gamma'_x$ with gate diameters $\mathcal O(r)$,
uniformly in $x$.
For every gate $V$ in $\mathsf\Gamma'_x$, we now consider
the gate $V'=\pi_x(\gamma_x^{-1}(V))$.
The corresponding circuit representation $\mathsf\Gamma_x$
implements the FDQC $\pi_x\circ\gamma_x\circ\pi_x^{-1}$,
with the same depth and gate support properties as $\mathsf\Gamma'_x$.
We then put
\[
\mathsf\Theta_x=\mathsf\Gamma_x\circ\mathsf C_x\circ\mathsf A'_x
\]
which has circuit range $\mathcal O(r)$, uniformly in $x$. All  gates except those of $\mathsf A'_x$ are supported in the boundary strip inside $L_x$, so the asserted support and agreement with $\mathsf A|_{L_{x-r}}$ are preserved. }
\end{proof}

\subsection{Injectivity of the index map in 2d}

In this section we use the above results, particularly the concept of symmetric blending, to establish the injectivity of the index map defined in Proposition \ref{prop.index}.

\begin{lemma} \label{lem:injectiveness}
    The index map $\ind: \SE_{G, 2} \longrightarrow H^3(G,U(1))$ is injective.
\end{lemma}

\begin{proof}
    {Let $(\alpha,\sfU,\AA)$ have trivial index, and retain the circuit representation $\sfA$ and the scale $r\geq\max\{nq,1\}$ fixed above. We write $[a,b]$ for the vertical strip $([a,b]\cap\ZZ)\times\ZZ$, and put $x_k=Nrk$, with $N$ introduced in Section \ref{sec: blending map} and to be chosen below.} 

    {Then Lemma \ref{lem:symmetric_blending} gives blendings $\theta_{x_k}$ with circuit representatives $\sfTheta_{x_k}$ of range $\mathcal O(r)$, agreeing with $\sfA$ on $L_{x_k-r}$.} 
    Define symmetric entanglers
    \[
    {\left(\alpha_k=\theta_{x_k}^{-1}\circ\theta_{x_{k+1}},
       \sfU\otimes\sfD\otimes\sfU',\widetilde\AA\right).}
    \]

    {Their circuit representatives
    $\sfTheta_{x_k}^{\rev}\circ\sfTheta_{x_{k+1}}$ have range
    $\mathcal O(r)$. Cancelling the common gates on the left,
    using $\sfTheta_x|_{L_{x-r}}=\sfA|_{L_{x-r}}$, leaves
    representatives $\sfA_k$ supported on
    $[x_k-Cr,x_{k+1}]$ and agreeing with $\sfA$ on
    $[x_k+Cr,x_{k+1}-Cr]$, for some $C\geq1$ that is   independent of $k$ and $N$. The circuit representatives $\sfA_k$ for $k \in 2\ZZ$ have disjoint supports provided that $N>C$.}

    It follows that
    \[
    {\left(\alpha_{\rm even}=\prod_{k\in2\ZZ}\alpha_k,
       \sfU\otimes\sfD\otimes\sfU',\widetilde\AA\right)}
    \]
    is a product of 1d FDQCs supported on disjoint strips.
    {A circuit representative $\sfA_{\rm even}$ is
    obtained by taking the union of corresponding layers of
    the $\sfA_k$, $k\in2\ZZ$, and has range $\mathcal O(r)$.
    It agrees with $\sfA$ away from boundary margins of width
    $\mathcal O(r)$ in each even strip.}

    The automorphism
    \[
    {\left(\alpha_{\rm odd}=(\alpha\otimes\id)
       \circ\alpha_{\rm even}^{-1},
       \sfU\otimes\sfD\otimes\sfU',\widetilde\AA\right)}
    \]
    {has circuit representative
    $\sfA\circ\sfA_{\rm even}^{\rev}$ of range $\mathcal O(r)$.

    Cancelling its common gates in the interiors of the even
    strips leaves a circuit $\sfA_{\rm odd}$ with no gates
    meeting $[x_k+C'r,x_{k+1}-C'r]$, for $k\in2\ZZ$ and some $C'\geq 1$ independent of $k$ and $N$. The remaining gates therefore lie
    in the strips $[x_{k-1}-C'r,x_k+C'r]$, $k\in2\ZZ$. 
    Finally we can then fix $N$ such that both $\alpha_{\rm even}$ and $\alpha_{\rm odd}$ are products of  1d FDQCs on disjoint strips (and such that they have circuit representatives supported on disjoint strips).}

    Since both $\alpha_{\rm even}$ and $\alpha_{\rm odd}$
    are symmetric FDQCs, they are striped symmetric entanglers
    in the sense of Section \ref{subsec:charges in 2D}.

    We now have 
        \begin{equation}\label{eq:stripes 2d-original}
            \alpha \otimes \id =\alpha_{\text{ odd}} \circ \alpha_{\text{even}},
    \end{equation}
    where each factor $\alpha_{\text{even}}$, $\alpha_{\text{odd}}$ is stably equivalent to $\id$ by Lemma \ref{lem.triviality-striped-symmetric-entanglers} (any two striped symmetric entanglers are stably equivalent). We conclude that $\alpha \otimes \id$, and therefore $\al$ itself, is stably equivalent to $\id$.

    To finish the proof, suppose $\al$ and $\al'$ are two symmetric entanglers with $[\mu] = \ind(\al) = \ind(\al') \in H^3(G, U(1))$. By Lemma \ref{lem.surjectivity-2d} there is a symmetric entangler $\delta$ with $\ind(\delta) = -[\mu]$. By additivity of the index (Proposition \ref{prop.index}) we have that $\delta \otimes \al'$ has trivial index and is therefore stably equivalent to $\id$. It follows that $\al \sim \al \otimes \id \sim \al \otimes \delta \otimes \al'$. But also $\al \otimes \delta$ has trivial index, and is therefore stably equivalent to $\id$. We therefore obtain $\al \sim \id \otimes \al' \sim \al'$, as required.
\end{proof}

\subsection{Proof of Proposition \ref{prop.classification-of-symmetric-entanglers}} \label{sec.proof-classification-symmetric-entanglers-2d}

\begin{proof}
    The case $d=0$ is elementary and is discussed in Lemma \ref{lem.classification of 0d symmetric entanglers}. The case $d=1$ is covered by Lemma \ref{lem.classification-1d-symmetric-entanglers}. For the case $d=2$, Proposition \ref{prop.index} yields a well-defined index map $\ind: \SE_{G, 2} \longrightarrow H^3(G,U(1))$ which is a homomorphism of monoids. Lemmas \ref{lem:injectiveness} and \ref{lem.surjectivity-2d} show that this index map is a bijection. Since $H^3(G,U(1))$ is a group, it follows that $\SE_{G, 2}$ is a group as well, and the index map is an isomorphism of groups.
\end{proof}

\section{Classification of 2d symmetry protected states with symmetric entanglers} \label{sec.classification of 2d SPT with SE}

In this section we prove Theorem \ref{thm.main}. The idea is to show that the index of symmetric entanglers descends to a well defined index of $G$-states with symmetric entanglers through $\ind \, \psi_{\al} := \ind \, \al$, and that this further descends to an isomorphism $\ind : \sSPT_{G, 2} \to H^3(G, U(1))$.

We proceed in three steps. Lemma \ref{lem:SE with invariant state is trivial} shows that an entangler fixing a product state has trivial index. With this, Lemma \ref{lem:well-defined index for G-states} shows that $\ind\psi_\alpha:=\ind\alpha$ is well-defined. Lemma \ref{lem.injective-spt} shows that stably equivalent entanglers prepare stably equivalent states, which gives injectivity.

\begin{lemma} \label{lem:SE with invariant state is trivial}
    Let $(\al,\sfU,\AA)$ be a symmetric entangler and $(\psi_0,\sfU,\AA)$ a product $G$-state such that $\psi_0 = \psi_0 \circ \al$. Then $\ind \, \al = [1]$.
\end{lemma}

\begin{proof}
    Let ($\beta,\caP_x^{[r]}$) be the LPS arising at the boundary algebra for $\al$ (see Lemma \ref{lem.LPS-boundary-algebra}).

    Note that $\psi_0$ restricts to the boundary algebra $\caP_x^{[r]}$ as a pure state {that we call $\psi$. Indeed,
    $$ \al(\caA_{\Left {x}}) = \caA_{\Left {x-r}} \otimes \caP_x^{[r]}.$$
    so the claim follows from the fact that $\psi_0$ restricts to a pure state on $\al(\caA_{\Left x})$ and on $\caA_{\Left {x-r}}$. In a similar way, the invariance of $\psi_0$ under the LPS $\beta^{\sfU}$ implies that $\psi$ is invariant under $\beta$.}  Even though $\psi_0$ is a product state on $\caA$, the state $\psi$ need not be a product state on $\caP_x^{[r]}$. However, it does have vanishing correlations beyond range $2r$.  That is, for $A \in (\caP_{x}^{[r]})_I$ and $B \in (\caP_{x}^{[r]})_J$ for $I, J \subset \Z$ with $\dist(I, J) > 2r$, we have
    \begin{equation} \label{eq:boundary state finite range correlations}
        \psi(A B) = \psi(A) \psi(B).
    \end{equation}
    Indeed, $A = \eta_x(a)$ and $B=\eta_x(b)$ for observables $a\in(\AA^{[r]}_{x})_I$ and $b\in (\AA^{[r]}_{x})_J$ (recall the definitions in Sect. \ref{subsubsec:LPS on P}). Since the FDQC $\eta_x$ has range at most $r$, the operators $A$ and $B$ have no support overlap as observables in $\AA$, where $\psi_0$ acts as a product state.

    We now show that $\psi$ satisfies the split property \cite{matsui.splitproperty}. Let $\phi = \psi_L \otimes \psi_R$ where $\psi_L$ is the restriction of $\psi$ to $\caP_L = \bigotimes_{y < 0} \caP_{x, y}^{[r]}$, and $\psi_R$ is the restriction of $\psi$ to $\caP_R = \bigotimes_{y \geq 0} \caP_{x, y}^{[r]}$. It is sufficent to show that $\psi$ and $\phi$ are quasi-equivalent. Any finitely supported observable $A \in \caP$ supported outside the interval $[-r, r]$ has a Schmidt decomposition (w.r.t. the Hilbert-Schmidt inner product) $A = \sum_{\kappa} a_{\kappa} A^{(\kappa)}_L A^{(\kappa)}_R$, with the $A^{(\kappa)}_L$ supported on $\{ y < -r \}$ and the $A^{(\kappa)}_R$ supported on $\{ y > r \}$. It then follows from Eq. \eqref{eq:boundary state finite range correlations} that $\psi(A) = \phi(A)$ for all such $A$. By density, $\psi$ and $\phi$ agree on ${\caP}_{[-r, r]^c}$ and therefore $\psi$ and $\phi$ are quasi-equivalent by \cite[~Lemma 2]{bratteli1978ground}. That is, $\psi$ satisfies the split property.

    It now follows from Lemma \ref{lem:anomaly is obstruction to gapped boundary} that $\ind \, \al = \Omega(\beta) = [1]$.
\end{proof}

\begin{lemma} \label{lem:well-defined index for G-states}
     The map
    \begin{equation}
\label{eq:index for SPT states}
\begin{alignedat}{2}
    \ind : {}\quad & \sSta_{G,2}              &\quad &\to \quad H^3(G,U(1)) \\
                   & (\psi_{\al},\sfU,\caA)  &      &\mapsto \quad \ind\,\al .
\end{alignedat}
\end{equation}
    is a well-defined index for $G$-states.
\end{lemma}

\begin{proof}
    We must check that $\ind (\psi_{\al}, \sfU, \caA) := \ind \, \al$ doesn't depend on the choice of product $G$-state and symmetric entangler. That is, if
    $$ \psi_{\al} = \phi_1 \circ \al_1 = \phi_2 \circ \al_2 $$
    for symmetric entanglers $\al_1, \al_2$ and product $G$-states $\phi_1, \phi_2$, then we must have $\ind \, \al_1 = \ind \, \al_2$.

    Lemma \ref{lem:alternative equivalence of G-product states} implies that there is a $G$-state $\tilde \phi$ on an auxiliary spin system $(\tilde \sfU, \widetilde \caA)$ and a {on-site} symmetric entangler $\gamma$ so that
    $$ \phi_2 \otimes \tilde \phi = (\phi_1 \otimes \tilde \phi) \circ \gamma.$$
    It follows that
    $$ (\phi_1 \otimes \tilde \phi) \circ (\al_1 \otimes \id) = (\phi_1 \otimes \tilde \phi) \circ \gamma \circ (\al_2 \otimes \id). $$
    In other words, the product $G$-state $\phi_1 \otimes \tilde \phi$ is invariant under the symmetric entangler $\gamma \circ (\al_2 \otimes \id) \circ (\al_1^{-1} \otimes \id)$, so Lemma \ref{lem:SE with invariant state is trivial} implies
    $$ \ind \, \gamma \circ (\al_2 \otimes \id) \circ (\al_1^{-1} \otimes \id) = [1]. $$
    {Since $\gamma$ is on-site, the boundary construction in Section~\ref{sec.index} can be performed with $\eta_x=\id$, so the resulting boundary LPS is on-site. Its anomaly is therefore trivial by Theorem~\ref{thm.main-classification-LPS}, and hence $\ind\gamma=[1]$.} This together with Lemma \ref{lem:composition is equivalent to stacking} and the fact that $\ind : \SE_{G, 2} \to H^3(G, U(1))$ is a group isomorphism (Proposition \ref{prop.classification-of-symmetric-entanglers}) implies that
    $$ \ind \, \al_1 = \ind \, \al_2, $$
    as required.
\end{proof}

The following lemma implies that the index defined in Lemma \ref{lem:well-defined index for G-states} is injective.

\begin{lemma} \label{lem.injective-spt}
    If $(\alpha,\sfU,\AA) \sim (\al',\sfU',\AA')$, then the $G$-states $(\psi_\alpha,\sfU,\AA)$ and $(\psi_{\al'}, \sfU', \AA')$ are stably equivalent.
\end{lemma}

\begin{proof}
First fix a symmetric entangler $\alpha$. We claim that the stable class of $\phi\circ\alpha$ is independent of the product $G$-state $\phi$.
By Lemma~\ref{lem:alternative equivalence of G-product states} and Lemma~\ref{lem.onsite-stabilization}, for any two such states $\phi_1,\phi_2$ there are a common special product $G$-state $\rho$ and an FDQC $\delta$ with symmetric gates such that
\[
\phi_2\otimes\rho=(\phi_1\otimes\rho)\circ\delta.
\]
Then $
\widehat\delta=(\alpha^{-1}\otimes\id)\circ\delta\circ(\alpha\otimes\id)$  is also an FDQC with symmetric gates (its gates are $(\alpha^{-1}\otimes\id)(V)$ with $V$ the gates of $\delta$). 
Since
\[
(\phi_2\circ\alpha)\otimes\rho
=\bigl((\phi_1\circ\alpha)\otimes\rho\bigr)\circ\widehat\delta,
\]
the claim is proven.

Suppose next that $\alpha$ and $\alpha'$ are equivalent, so that
$\alpha'=\mathscr F\circ(\alpha\circ\gamma)\circ\mathscr F^{-1}$
for an equivariant spin system isomorphism $\mathscr F:(\sfU,\AA)\to(\sfU',\AA')$ and an FDQC $\gamma$ with symmetric gates. Writing $\psi_{\alpha'}=\phi'\circ\alpha'$, we have
\[
    \psi_{\alpha'}\circ\mathscr{F}
    =(\phi'\circ\mathscr{F})\circ\alpha\circ\gamma.
\]
Since $\phi'\circ\mathscr{F}$ is a product $G$-state on $\AA$,
the preceding claim shows that
$(\phi'\circ\mathscr{F})\circ\alpha$ is stably equivalent to
$\psi_\alpha$. Composition with $\gamma$ and transport by
$\mathscr{F}$ preserve equivalence, so
$\psi_{\alpha'}$ and $\psi_\alpha$ are stably equivalent.

Finally, if $\alpha$ and $\alpha'$ are stably equivalent, choose auxiliary systems making their identity extensions equivalent. Stack each auxiliary system further with an on-site regular representation, so that it admits a special product $G$-state. The corresponding extensions of $\psi_\alpha$ and $\psi_{\alpha'}$ are prepared by equivalent entanglers, and hence are stably equivalent by the preceding section. As these extensions use special product $G$-states, the original states are stably equivalent as well.
\end{proof}

\noindent\begin{proofof}[Theorem \ref{thm.main}]
    From Lemma \ref{lem:well-defined index for G-states} we have a well-defined index map
    $$ \ind : \sSta_{G, 2} \to H^3(G, U(1)) : ( \psi_{\al}, \sfU, \caA ) \mapsto \ind \, \al. $$
    {As in the proof of Lemma \ref{lem.injective-spt}, stacking with on-site regular representations ensures that every stable entangler class has a representative admitting a special product $G$-state.}
    Proposition \ref{prop.classification-of-symmetric-entanglers} implies that this map descends to a surjective homomorphism
    \begin{equation} \label{eq:index of SPT phases}
    \ind : \sSPT_{G, 2} \to H^3(G, U(1)).
    \end{equation}
{To prove injectivity, suppose that two $G$-states
$\psi_\alpha$ and $\psi_{\alpha'}$ have the same index.
Then $\ind\,\alpha=\ind\,\alpha'$, so Proposition
\ref{prop.classification-of-symmetric-entanglers} implies
that the symmetric entanglers $\alpha$ and $\alpha'$
are stably equivalent. Lemma \ref{lem.injective-spt}
therefore implies that $\psi_\alpha$ and $\psi_{\alpha'}$
are stably equivalent.}

\end{proofof}

\appendix


\section{Group cohomology} \label{sec.cohomology}

For a comprehensive discussion on finite group cohomology, see e.g.  \cite{brown2012cohomology}. Let $G$ be a finite group. For any left G-module $M$, we denote by $C^n(G,M)$ the set of functions from $G \times \dots \times G \to M$, where the product runs over $n$ copies of $G$. Elements of $C^n(G,M)$ are called $n$-cochains. The pair $(C^n(G,M),\cdot)$ is a group for every $n$, so that co-boundary group homomorphisms can be defined as
	\begin{align*}
		\nonumber    d^{n+1}:\ C^n(G,M)&\to C^{n+1}(G,M) \\
		\mu &\mapsto (d^{n+1} \mu),
	\end{align*}
	given explicitly by
	\begin{align} \label{eq.cocyclecondition}
		\nonumber (d^{n+1}\mu) (g_1,\dots,g_{n+1}) := &g_1\cdot \mu (g_2,\dots,g_{n+1}) \\ &+ \sum\limits_{i=1}^n (-1)^i \mu(g_1,\dots,g_ig_{i+1},\dots,g_{n+1}) + (-1)^{n+1} \mu(g_1,\dots,g_n),
	\end{align}
    where we are using additive notation for the abelian module operation. The map
    \begin{equation}
        d^{n+1} \circ d^n:\ C^{n-1}(G,M) \to C^{n+1}(G,M)
    \end{equation}
    is identically zero, from which we conclude there is a cochain complex
	$\dots \leftarrow C^{n} \leftarrow C^{n-1} \leftarrow \dots,$
	with $n$-th cohomology group defined as 
	\begin{equation}
		H^n(G,M) := \dfrac{Z^n(G,M)}{B^n(G,M)},
	\end{equation}
	with 
	\begin{align*}
		Z^n(G,M) := \ker(d^{n+1}), \\
		B^n:= \begin{cases} 0,\ &\text{if } n=0, \\
			\text{im}(d^n),\ &\text{if } n\ge 1.
		\end{cases}
	\end{align*}

    That is, the n-th cohomology group $H^n(G, M)$ is the group of equivalence classes of n-cocycles (elements of $Z^n(G,M)$) with respect to equivalence modulo n-coboundaries (elements of $B^n(G,M)$). \par


\subsection{Projective representations} \label{appendix.projective-reps}

We define a projective representation ($V, \mathcal K$) with 2-cocycle $\mu \in Z^2(G,U(1))$ as a finite-dimensional Hilbert space $\mathcal K$, together with a map 
$$V:\ G \longrightarrow \mathcal U(\mathcal K)$$
such that $V(g) V(h) = \mu(g,h) V(gh)$. Two projective representations $(V,\mathcal K)$, $(V',\mathcal K')$ are isomorphic if there exists a unitary intertwiner $\mathfrak U: \mathcal K \to \mathcal K'$ such that
$\mathfrak U \circ V(g) = V'(g) \circ \mathfrak U,$
for all $g \in G$. By \cite[~Corollary 3.8]{CHENG2015230}, two projective representations with the same 2-cocycle $\mu$ and the same character are isomorphic. This implies

\begin{lemma} \label{lem.isomorphism-projective-reps}
    Let $(V,\HH)$, $(V',\HH')$ be two projective representations of the same dimension with the same 2-cocycle $\mu$. Then 
    $(V \otimes L, \, \HH \otimes \mathbb C[G])$ is isomorphic to $ (V' \otimes L, \, \HH' \otimes \C[G])$,
    {where $L$ is the left regular representation on $\mathbb C[G]$,
defined by $L(g)\ket{h}=\ket{gh}$ for $g,h\in G$.} 
\end{lemma}
\begin{proof}
    Let $\chi_V,\chi_{V'}, \chi_L$ denote the character functions of projective representations $V$, $V'$, and $L$, respectively.  Since $\chi_L (g) = {|G|}\delta_{e,g}$, the result follows by $
        \chi_{V \otimes L}(g) = \chi_V (g) \cdot \chi_L (g) = {|G|}{\mu(e,e)}\dim (V) \delta_{e,g} = \chi_{V'\otimes L}(g).$
\end{proof}

\begin{lemma}[Classification of 0d LPS] \label{lem.classification-0d-LPS}
    There is a group isomorphism
    \begin{align*}
       \Omega:  (\Sym_{G,0} / \sim) \to H^2(G,U(1)). 
    \end{align*}
\end{lemma}
\begin{proof}
    A 0-dimensional locality preserving symmetry is a pair $(\beta,\AA)$, where the algebra $\AA = \End{(\HH)}$ is the full matrix algebra on a finite-dimensional Hilbert space $\HH$, and the symmetry action $\beta_g = \Ad{U(g)}$ for a unitary projective representation $g \to U(g) \in \AA$ of $G$. 

    The anomaly index $\Omega(\beta)$ of $(\beta,\AA)$ is defined as the projective class of $g \to U(g)$, and takes value in $H^2(G,U(1))$. For any 2-cocycle $\mu$, consider $\AA = \End(\C[G])$ and $\beta$ given by $g \mapsto U(g)$ with $U(g) \ket{h} = { \mu(g, h) } \ket{gh}$. Then $\Omega( \beta ) = [\mu] \in H^2(G, U(1))$. Thus the index map is surjective. 

    Suppose now $(\beta, \AA)$ and $(\beta', \AA')$ are 0-dimensional LPS with $\Omega(\beta) = \Omega(\beta') = [\mu]$. By stacking with trivial representations, we may assume w.l.o.g. that $\caA$ and $\caA'$ have the same dimension. Moreover, possibly by redefining phases, one can pick the projective representations so that they yield the same 2-cocycles (not just up to coboundaries). The stable equivalence $(\beta, \AA) \sim (\beta', \AA')$ now follows from Lemma \ref{lem.isomorphism-projective-reps}, showing injectivity.
  \end{proof}

\section{Locality preserving symmetries} \label{appendix.LPS}

\subsection{The anomaly of a locality preserving symmetry} \label{sec:anomaly}

The idea behind the definition of the anomaly presented here goes back to \cite{Else_2014}, and we present the definition discussed in \cite{bols2025classificationlocalitypreservingsymmetries}. {The component QCAs of any locality preserving symmetry
on a spin chain are FDQCs, due to the finiteness of $G$
\cite{zhang2023topological}.
By \cite[Lemma 3.1]{bols2025classificationlocalitypreservingsymmetries},
a one-dimensional FDQC of automorphism range $R$ admits
a depth-two circuit representation with gate diameters
at most $2R$.}

Let $\beta$ be an LPS of range $R$. A \emph{right restriction} $\beta^{\geq j}$ of $\beta$ at $j \in \Z$ with \emph{defect size} $L$ is a family of automorphisms $\beta^{\geq j}_{g}$ such that for any $g\in G$ 
$$
\beta^{\geq j}_{g} |_{\caA_{\geq (j + L)}} = \beta_g|_{\caA_{\geq (j + L)}} \, \text{ and } \,\, \beta^{\geq j}_{g}|_{\caA_{< j-L}} = \id_{\caA_{< j-L}}.
$$
By the above considerations, any symmetry of range $R$ admits right restrictions at all sites with defect size $2R$.

Given a right restriction $\beta^{\geq j}$ of defect size $L$, there are local unitaries $\Fusion_{j}(g, h) \in \caA_{[j - L, j + L + R]}$, called \emph{fusion operators} associated to $\beta^{\geq j}$, such that
$$
\beta^{\geq j}_{g} \circ \beta^{\geq j}_{h} = \Ad{\Fusion_{j}(g, h)} \circ \beta^{\geq j}_{gh}.
$$
These unitaries are uniquely determined by this equation up to phase. They capture the failure of $g \mapsto \beta^{\geq j}_{g}$ to be a group homomorphism.

Using associativity to compute $\beta^{\geq j}_{f} \circ \beta^{\geq j}_{g} \circ \beta^{\geq j}_{h}$ in two different ways one obtains that there are phases $\omega_{j}(f, g, h) \in U(1)$ such that
\begin{equation} \label{eq-fusion.operators}
\Fusion_{j}(f, g) \Fusion_{j}(fg, h) = \omega_{j}(f, g, h) \times \beta^{\geq j}_{f} \big( \Fusion_{j}(g, h) \big) \Fusion_{j}(f, gh)
\end{equation}
for all $f, g, h \in G$.

We now state the main classification result of LPS on spin chains proven in \cite{bols2025classificationlocalitypreservingsymmetries}, see also \cite{seifnashri2025disentangling}. The notion of equivalence used in that work is slightly different from the equivalence relation $\sim$ on $\Sym_{G, d}$ introduced in this paper (Section \ref{subsec:classification-lps}). We say two LPSs $(\beta, \caA)$ and $(\beta', \caA')$ are \emph{projectively stably equivalent}, denoted as
$$(\beta,\AA) \sim_{\proj} (\beta', \AA'),$$
if there exist {on-site} symmetries $(\delta, \tilde \AA)$ and $(\delta',\tilde \AA')$ such that $(\beta \otimes \delta, \AA \otimes \tilde \AA) \text{ and } (\beta' \otimes \delta', \AA' \otimes \tilde \AA')$
are equivalent.

This differs from stable equivalence $\sim$ in that {the stacked symmetries are arbitrary on-site LPSs: their on-site implementers may be projective, whereas $\sim$ requires auxiliary actions induced by genuine unitary representation data}.

\begin{theorem} [Theorem 2.1 and Proposition 3.2. of \cite{bols2025classificationlocalitypreservingsymmetries}] \label{thm.main-classification-LPS}
    The map $\omega_{j} : G^3 \rightarrow U(1)$ is a 3-cocycle, and the corresponding group cohomology class $[\omega_{j}] \in H^3(G, U(1))$ depends only on the symmetry $\beta$, i.e. the cohomology class is independent of the site $j$ and the choice of right restriction $\beta^{\geq j}$. We thus obtain a well defined map $$\ano : \Sym_{G,1} \rightarrow H^3(G, U(1))$$ which is called \emph{the anomaly}.
    The anomaly lifts to an isomorphism of groups
    $$\ano : (\Sym_{G,1} / \sim_{\proj}) \rightarrow H^3(G, U(1)).$$
\end{theorem}

\subsection{Disentangling anomaly-free LPSs}

This subsection is devoted to showing that the difference between stable equivalence $\sim$ and projective stable equivalence $\sim_{\proj}$ actually does not matter for the classification of LPSs on spin chains.

\begin{lemma} \label{lem.trivializing-LPS}
    If an LPS $(\beta, \AA)$ on a spin chain $\AA$ has a trivial anomaly $\Omega(\beta) = [1]$, then it is stably equivalent to {an action induced by an on-site unitary representation datum}.  {Let $r\geq 1$ be a bound on the range of $\beta$. Then the dimensions of the auxiliary on-site algebras are uniformly bounded by a function of $r$, $G$, and $\sup_j \, \dim \caA_j$. The FDQC that disentangles the LPS $\beta$ has a (circuit) range bounded by $\mathcal O(r)$.} 
\end{lemma}

\begin{proof}
    By \cite[Proposition 6.2]{bols2025classificationlocalitypreservingsymmetries}, the statement holds without mention of linearity of the on-site actions. That is, there is a {on-site} LPS $(\delta, \caD)$ and a disentangling FDQC $(\gamma, \caA \otimes \caD)$ whose range is bounded by $\mathcal{O}(r)$ so that $\gamma \circ (\beta \otimes \delta) \circ \gamma^{-1} 
    $
    is an {on-site} LPS acting projectively on each site. {By inspection the proof of \cite[Proposition 6.2]{bols2025classificationlocalitypreservingsymmetries}, one checks that the auxiliary on-site dimension $\sup_j\dim(\caD_j)$ is bounded in terms of $r$, $G$, and $\sup_j \, \dim \caA_j$.}
    
    Let $\overline \delta$ denote the symmetry with conjugate on-site projective actions. Then $(\delta \otimes \overline \delta, \caD \otimes \caD)$ {is induced by an {on-site unitary representation datum}}, and  
    $$
    \left(\gamma \otimes \id \right)\circ (\beta \otimes \delta \otimes \overline \delta) \circ\left( \gamma^{-1}  \otimes \id\right) = \left(\gamma \circ (\beta \otimes \delta) \circ \gamma^{-1}\right) \otimes \overline \delta
    $$
    is {on-site}. Hence $(\beta,\AA)$ is stably equivalent to a {on-site} LPS. By Lemma \ref{lem.equivalence-decoupled-1d}, any two 1d {on-site} LPSs are stably equivalent with on-site ancillary dimension bounded in terms of $G$ and the local dimension bounds of the initial spin systems. Thus $(\beta,\AA)$ is stably equivalent to {an action induced by an {on-site unitary representation datum}}, through stacking with auxiliary spin chains whose local dimensions are bounded as claimed.
\end{proof}

By Lemma \ref{lem.trivializing-LPS}, stable equivalence with respect to {stacking with actions induced by {on-site unitary representation data}} is equivalent to stable equivalence with respect to stacking with {on-site} symmetries. This result together with Theorem \ref{thm.main-classification-LPS} yields:

\begin{proposition} [Classification of 1d LPSs] \label{thm.classification-1d-LPS}
    There is a group isomorphism
    \begin{align*}
       \Omega:  (\Sym_{G,1} / \sim) \to H^3(G,U(1)). 
    \end{align*}
\end{proposition}

\subsection{Invariant split states and trivial anomaly}

{Recall that a pure state $\psi$ on a spin $\AA$ satisfies the split property \cite{matsui.splitproperty} if for any (and therefore all) $j \in \Z$, the state $\psi$ is {unitary equivalent} to a factorized state $\psi_{<j} \otimes \psi_{\ge j}$, where $\psi_{<j}$ and $\psi_{\ge j}$ are pure states on $\AA_{<j}$ and $\AA_{\ge j}$, respectively.  See \cite{ogata2019classification} for equivalent definitions. The following result was also proved in, e.g. \cite{KapustinSopenko2025}.}

\begin{lemma} \label{lem:anomaly is obstruction to gapped boundary}
    Let $(\psi,\AA)$ be a pure state on a spin chain which satisfies the split property. Assume that $\psi$ is invariant with respect to an LPS $(\beta,\AA)$. Then $\Omega(\beta) = [1]$. 
\end{lemma}
\begin{proof}
By the discussion in the first paragraph of Section \ref{sec:anomaly}, each QCA component of $\beta$ has an FDQC representation. Consequently, for each $g\in G$, there exist a unitary $V_g^j \in \AA$ and (left, right) restrictions $\beta_g^{< j},\beta^{\ge j}_g$ such that 
$$\beta_g = (\beta_g^{< j} \otimes \beta^{\ge j}_g) \circ \Ad{V_g^j}.$$ 
    By the split property, $\psi \simeq \psi_{< j} \otimes \psi_{\ge j}$, for a pure product state $\psi_{< j} \otimes \psi_{\ge j}$, where $\simeq$ means unitary equivalence. Hence
    \begin{align*}
        \psi &= \psi \circ \beta_g \\
        &\simeq \psi \circ (\beta_g^{< j} \otimes \beta_g^{\ge j}) \\
        &\simeq (\psi_{< j} \circ \beta_g^{< j}) \otimes (\psi_{\ge j} \circ \beta_g^{\ge j})
    \end{align*}
  {  Since also $\psi \simeq \psi_{< j} \otimes \psi_{\ge j}$, we deduce that $\psi_{< j} \simeq \psi_{< j} \circ \beta_g^{< j}$ and $\psi_{\ge j} \simeq \psi_{\ge j} \circ \beta_g^{\ge j}$. }
 This {means} that the left (and right) restrictions of $\beta_g$ are implementable in the GNS representation $(\pi_{\psi}, \HH_{\psi}, \Psi)$ of $\psi$. That is, there exist unitaries $g\to u_g \in \caB(\HH_{\psi})$ such that $\pi_\psi\circ \beta_g^{\ge j}(A) = \Ad{u_g} (  \pi_{\psi}(A)  )$ for every $A \in \AA$. Irreducibility of $\pi_{\psi}$ then implies that 
    $$\lambda(g,h) u_gu_hu_{gh}^{-1} = \pi_\psi(\Fusion_j(g,h)),$$
    for a 2-cochain ${\lambda(g,h)} \in U(1)$. A simple computation now yields 
    {\begin{align*}
        \omega(f,g,h) = (\d \lambda^{-1})(f,g,h),
    \end{align*}}
    so $\Omega(\beta) = [\omega] = [1]$.
\end{proof}

\subsection{Stable equivalence and LPS implemented by QCAs}

We have defined (see \ref{subsec:classification-lps}) equivalence of LPSs with respect to conjugation with FDQCs. By stabilization, this equivalence can be lifted to accommodate also QCAs, as stated in the following Lemma:

\begin{lemma} \label{lem.equivalence-symmetries-qca}
    Let $(\beta,\AA)$ be a LPS on a spin chain $\AA$, and $(\gamma,\AA)$ a QCA on $\AA$. Then $$(\beta,\AA) \sim (\gamma \circ \beta \circ \gamma^{-1}, \AA).$$
\end{lemma}
\begin{proof}
    By definition,
    \begin{align*}
    (\gamma \circ \beta \circ \gamma^{-1}, \AA) &\sim \left(\left(\gamma \circ \beta \circ \gamma^{-1}\right)\otimes \id, \AA \otimes \AA\right) \\ &= \left(\left(\gamma \otimes \gamma^{-1}\right) \circ \left(\beta \otimes \id\right) \circ \left(\gamma^{-1} \otimes \gamma\right), \AA \otimes \AA  \right).
    \end{align*}
    But $(\gamma \otimes \gamma^{-1},\AA \otimes \AA)$ is an FDQC \cite{Gross_2012}, so the result follows by definition of stable equivalence. 
\end{proof}



\section{Eilenberg–Mazur swindle}\label{appendix.swindle}

In Lemma \ref{lem:injectiveness} we prove injectivity of the $H^3(G, U(1))$-index by showing that a symmetric entangler with trivial index is stably equivalent to $\id$. We do this by first showing the entangler is stably equivalent to a composition of entanglers each of which acts on decoupled vertical strips. Each of these strips is a 1D symmetric entangler, and therefore has a possibly non-trivial $H^2(G, U(1))$-valued anomaly. It is well known that these anomalies are not obstructions to further disentangling the 2D circuit because they can be `pumped to infinity'. This is an instance of the so-called Eilenberg-Mazur swindle\footnote{The Eilenberg-Mazur swindle is a technique originating in algebraic topology \cite{mazur1961embeddings}, attributed to Samuel Eilenberg in \cite{bass1963big}. The adoption of the name for the analogous technique of pumping charges to infinity in the context of topological phases of matter seems to originate with \cite{freedman2020classification}.}, which we explain in this appendix.


\subsection{$H^2(G,U(1))$ charges in 2d} \label{subsec:charges in 2D}

Let $\caA = \otimes_{j \in \Z^2} \caA_j$ be a two-dimensional spin system.
Let $\{ I_i \}_{i \in \Z}$ be a disjoint family of discrete intervals of uniformly bounded size. Let $\alpha = \otimes_{i} \alpha_{I_i}$ where the ${\al_{I_i}}$ are FDQCs on $\caA$ of uniformly bounded range, whose gates are supported on the vertical strips ${v_i = I_i \times \Z}$. Any FDQC of this form is called a \emph{striped FDQC} of $\caA$. If the striped FDQC is moreover equivariant w.r.t. {the action of a specified {on-site unitary representation datum}}, we call it a \emph{striped symmetric entangler}.

\begin{lemma} \label{lem.triviality-striped-symmetric-entanglers}
    Any two striped symmetric entanglers in 2d are stably equivalent.  {The on-site dimensions of the auxiliary spin systems needed to realize this equivalence are uniformly bounded in terms of $G$, the original on-site algebra dimension bound, a common circuit-range bound for the $\alpha_{I_i}$, and the maximum strip width.}
\end{lemma}

\begin{proof}
    
We prove that a striped symmetric entangler $\alpha = \otimes_i \alpha_{I_i}$ is stably equivalent to the identity. Since the stripes $I_i$ are disjoint, each $\alpha_{I_i}$ is $G$-equivariant. {Whenever necessary, insert strips of width one carrying the identity, so that the intervals partition $\ZZ$, and label them in spatial order.} {By regrouping sites, we may assume that each strip is a vertical line and write $\alpha_{I_i}=\alpha_i$. By \cite{Gross_2012}, a one-dimensional FDQC of range $R$ admits a depth-two circuit representation with gate diameters $\mathcal O(R)$. Since the ranges of the $\alpha_i$ are uniformly bounded, these circuit representations can be chosen with circuit ranges uniformly bounded in $i$. By Lemma \ref{lem.classification-1d-symmetric-entanglers}, there is a well-defined complete invariant $\ind(\alpha_i)\in H^2(G,U(1))$ for each $\alpha_i$, viewed as a one-dimensional symmetric entangler on its support.}
Lemma \ref{lem.surjectivity-1d} provides for each class $[\omega] \in H^2(G,U(1))$ a 1-dimensional symmetric entangler $(\alpha^{[\omega]}, \sfU^{[\omega]}, \AA^{[\omega]})$ with $\ind(\alpha^{[\omega]}) = [\omega]$. Since $H^2(G, U(1))$ is finite, this set of entanglers is finite, and therefore has {on-site algebra dimensions and ranges bounded in terms of $G$}.

We define an auxiliary striped symmetric entangler by
\begin{align} \label{eq.enriched-symmetric-entangler-swindle}
    (\tilde \alpha, \tilde \sfU, \tilde \AA) := (\otimes_{i \geq 1} \alpha^{[\omega_i]}, \otimes_{i \geq 1} \sfU^{[\omega_i]}, \otimes_{i \geq 1} \AA^{[\omega_i]}),
\end{align}
where
\begin{enumerate}
    \item for odd $i \in \NN$,
    \begin{align} \label{eq.def-charges-odd}
        [\omega_i] = \ind(\alpha_{i}) + \ind(\alpha_{i-1}) + \dots + \ind(\alpha_{1})
    \end{align}
    \item for even $i \in \NN$,
    \begin{align} \label{eq.def-charges-even}
        [\omega_i] = -[\omega_{i-1}] = -\ind(\alpha_{i-1}) - \ind(\alpha_{i-2}) - \dots - \ind(\alpha_1).
    \end{align}
\end{enumerate}
We have for any 2-cocycle $\omega$ that $\alpha^{[\omega]} \otimes {(\alpha^{[\omega]})^{-1}}$ is stably equivalent to $\id$ as a 1-dimensional symmetric entangler. Define another auxiliary symmetric entangler $\tilde \al'$ by
$$
(\tilde \alpha', \tilde \sfU, \tilde \AA) := (\otimes_{i \geq 1} {(\alpha^{[\omega_i]})^{-1}}, \otimes_{i \geq 1} \sfU^{[\omega_i]}, \otimes_{i \geq 1} \AA^{[\omega_i]}),
$$
then the symmetric entangler $(\tilde \alpha \otimes \tilde \alpha', \tilde \sfU \otimes \tilde \sfU, \tilde \AA \otimes \tilde \AA  )$ is stably equivalent to $(\id \otimes \id, \tilde \sfU \otimes \tilde \sfU, \tilde \AA \otimes \tilde \AA)$. We therefore find
\begin{align*}
    ( \alpha , \sfU, \AA ) \sim (\alpha \otimes \tilde \alpha \otimes \tilde \alpha', \sfU \otimes \tilde \sfU \otimes \tilde \sfU, \AA \otimes\tilde \AA \otimes \tilde \AA).
\end{align*}
The right-hand side is still a striped symmetric entangler, and each strip $i \geq 1$ now comes with index $\ind(\alpha_i) + [\omega_i] - [\omega_i] \in H^2(G,U(1))$. We represent the obtained symmetric entangler on the right half-plane by
\begin{align*}
    \begin{bmatrix}
        -[\omega_{1}] \\ 
        [\omega_{1}] \\
        \ind(\alpha_{1})
    \end{bmatrix}
    \qquad \otimes \qquad 
    \begin{bmatrix}
        -[\omega_2] \\ 
        [\omega_2] \\
        \ind(\alpha_2)
    \end{bmatrix}
    \qquad \otimes \qquad 
    \begin{bmatrix}
        -[\omega_{3}] \\ 
        [\omega_{3}] \\
        \ind(\alpha_{3})
    \end{bmatrix}
    \qquad 
    \dots 
\end{align*}
where each tensor factor represents a strip, with the corresponding stacked structure as described above. 

Let $i \in 2\NN$ be even. The strips $i,i+1$ carry indices

\begin{align*}
    \begin{bmatrix}
        -[\omega_i] \\ 
        [\omega_i] \\
        \ind(\alpha_i)
    \end{bmatrix}
    \qquad \otimes \qquad 
    \begin{bmatrix}
        -[\omega_{i+1}] \\ 
        [\omega_{i+1}] \\
        \ind(\alpha_{i+1})
    \end{bmatrix}.
\end{align*}
From definitions \eqref{eq.def-charges-even} and \eqref{eq.def-charges-odd} we get
$$\ind(\alpha_i) + \ind(\alpha_{i+1}) - [\omega_i] - [\omega_{i+1}]=[1],$$
so by Lemma \ref{lem:injectiveness-1d} we find that the joint strip $\{i,i+1\}$ of $\al\otimes\tilde\al'$ is stably equivalent to identity as it is a symmetric entangler with trivial anomaly. Moreover, in constructing the stably equivalence relation of an entangler with trivial anomaly in the proof of Lemma \ref{lem:injectiveness-1d}, we have used symmetric blending (Lemma \ref{lem.symmetric-blending-1d}) and thus have stacked on-site with at most the boundary algebra and once the regular representation. {{The dimension of each boundary algebra is bounded by $D_{\mathrm{eff}}^{C\rho}$, where $D_{\mathrm{eff}}$ bounds the on-site algebra dimensions after regrouping and stacking with the model entanglers, $\rho\geq1$ is a common range bound for the resulting entanglers, and $C$ is an absolute constant. These bounds depend only on $G$, the original on-site dimension bound, the common circuit-range bound and the maximal strip width.} {This gives the asserted bound on the auxiliary on-site dimension.}
{The subsequent application of Lemma \ref{lem.equivalence-block-partitioned-1d} preserves this bound in terms of the same parameters.}}

{Hence, the full joint strip $\{i,i+1\}$ of $\al\otimes\tilde\al\otimes\tilde\al'$ is stably equivalent to a pair of neighbouring strips $\al^{[\omega_i]} \otimes \al^{[\omega_{i+1}]}$. 
{The bound on the auxiliary on-site dimension is uniform over these pairs.}}
We therefore find that $\alpha \otimes \tilde \alpha \otimes \tilde \alpha'$ is stably equivalent to a striped 2-dimensional symmetric entangler with a strip charge configuration on the right half-plane as follows: 
\begin{equation*}
        [\omega_{1}]    \qquad \otimes \qquad [\omega_2] \qquad \otimes \qquad [\omega_{3}] \qquad \dots 
\end{equation*}

For each odd $i \in \NN$, we now have neighbouring stripes with indices:
\begin{equation*}
    [\omega_i] \qquad \otimes \qquad [\omega_{i+1}].
\end{equation*}
By definition (Eqs. \eqref{eq.def-charges-odd}, \eqref{eq.def-charges-even}) we have $[\omega_i] + [\omega_{i+1}] = 0$, so the entangler restricted to each joint strip $\{i,i+1\}$ is stably equivalent to an identity entangler by Lemma \ref{lem.classification-1d-symmetric-entanglers}.

{Again, the auxiliary on-site dimension satisfies the asserted bound, and} the FDQCs with symmetric gates realizing the equivalence with $\id$ have uniformly bounded ranges. From this, we conclude that $(\alpha, \sfU, \AA)$ is stably equivalent to its own restriction to the left half-plane, all the striped symmetric entanglers on the right half-plane having been removed by the above construction.

The same construction can now be repeated to remove all the remaining striped symmetric entanglers on the left half-plane. We conclude that $(\alpha, \sfU, \AA)$ is stably equivalent to the identity.
{{Stacking the auxiliary systems for two such equivalences to the identity gives the claimed equivalence between any two striped symmetric entanglers, with the asserted bound on the auxiliary on-site dimension.}}
\end{proof}

\subsection{Triviality of block partitioned symmetric entanglers}

Let $(\alpha, \sfU,\AA)$ be a block partitioned 1d symmetric entangler. Without loss of generality, assume that $\al$ is factorized. That is, $\alpha$ leaves each on-site algebra $\caA_i$ invariant and restricts there to an equivariant automorphism $\Ad{v_i}: \AA_i \to \AA_i$, for some unitary $v_i \in \caA_i$. Equivariance of $\Ad{v_i}$ means that
$$ \beta_g( v_i ) \beta_g(x) \beta_g(v_i^*) = v_i \beta_g(x) v_i^* $$
for all $x \in \caA_i$. This implies that there are phases $\lambda_i(g) \in U(1)$ such that $\beta_g( v_i ) = \lambda_i(g) v_i$. The fact that $g \mapsto \beta_g$ is a group homomorphism implies that $g \mapsto \lambda_i(g)$ is a 1-cocycle. That is, $v_i$ transforms as a one-dimensional representation $\lambda_i$ of $G$. One might think that these on-site ``charges'' are obstructions to deforming $\alpha$ to $\id$ via an FDQC with symmetric gates. We argue in Lemma \ref{lem.equivalence-block-partitioned-1d} below that this is not the case by showing that these charges can be pumped to infinity using a swindle argument.

The proof is completely analogous to the one of Lemma \ref{lem.triviality-striped-symmetric-entanglers}, except that the role of the 1-dimensional symmetric entanglers with $H^2(G, U(1))$-valued index will now be played by 0-dimensional symmetric entanglers with $H^1(G, U(1))$-valued index. Let us prove the required complete classification of such symmetric entanglers (the case $d=0$ of Proposition \ref{prop.classification-of-symmetric-entanglers}):
\begin{lemma} \label{lem.classification of 0d symmetric entanglers}
    The monoid $\SE_{G, 0}$ is a group and there is a group isomorphism
    $$ \ind : \SE_{G, 0} \simeq H^1(G, U(1)). $$
\end{lemma}

\begin{proof}
    A 0-dimensional symmetric entangler is a triple $(\al, \sfU, \caA)$ where the algebra $\caA = \End{(\HH)}$ is the full matrix algebra on a finite-dimensional Hilbert space $\HH$, {the symmetry datum $\sfU=(U)$ consists of a single genuine unitary representation $U:G\to\caU(\HH)$ with induced action $\beta_g=\Ad{U(g)}$}, and $\al = \Ad{v}$ for a unitary $v \in \caA$, unique up to phase, such that $\beta_g \circ \al = \al \circ \beta_g$ for all $g \in G$.

    The latter condition implies that $\beta_g(v) = \lambda(g) v$ for a 1-cocycle $\lambda : G \rightarrow U(1)$. The index is defined by
    $$ \ind(\al) := [\lambda] \in H^1(G, U(1)). $$
    Note that in $H^1(G, U(1))$, the equivalence relation defining classes is trivial, i.e. if $[\lambda] = [\mu] \in H^1(G, U(1))$ then $\lambda = \mu$ are equal 1-cocycles.

    Let us first argue that this index map is surjective. Given a 1-cocycle $\lambda$, let $\HH = \C[G] = \rm{span}\{ \ket{g} \, : g \in G \}$ and {let $\sfU=(U)$ consist of the regular representation and put $\beta_g=\Ad{U(g)}$}. Let $v \in \caU(\HH)$ be defined by $v \ket{h} = \lambda(h)^{-1} \ket{h}$ and $\al = \Ad{v}$. Then one easily checks that $\beta_g(v) = \lambda(g) v$. This proves surjectivity.

    Suppose now
    \[
    (\al^{(i)} = \Ad{v^{(i)}}, \sfU^{(i)} = (U^{(i)}), \caA^{(i)} = \End(\HH^{(i)})),\qquad i=1,2,
    \]
    are two 0-dimensional symmetric entanglers with the same index $\ind(\al^{(1)}) = \ind(\al^{(2)}) = [\lambda]$. That is, $\beta^{(i)}_g(v^{(i)}) = \lambda(g) v^{(i)}$ for all $g \in G$ and $i = 1, 2$.

    By stacking with trivial representations of $G$ of appropriate dimensions and extending $v \rightarrow v \otimes \I$ we can assume without loss of generality that $\HH^{(1)}$ and $\HH^{(2)}$ have the same dimension. Further stacking with the regular representation of $G$ and extending $v \rightarrow v \otimes \I$ we get from Lemma \ref{lem.isomorphism-projective-reps} that there is a unitary isomorphism $\Phi : \HH^{(1)} \otimes \C[G] \rightarrow \HH^{(2)} \otimes \C[G]$ {which intertwines the genuine representations $U^{(1)}\otimes U^{\rm{reg}}$ and $U^{(2)}\otimes U^{\rm{reg}}$}. This means that, after stacking, we may assume without loss of generality that our 0-dimensional symmetric entanglers act on the same Hilbert spaces $\HH^{(1)} = \HH^{(2)} = \HH$ {with the same symmetry data $\sfU^{(1)}=\sfU^{(2)}=\sfU$, and hence the same induced action $\beta$}.

    In that case, note that the unitary $V := v^{(2)} (v^{(1)})^*$ is symmetric : $\beta_g(V) = V$ for all $g \in G$. Moreover, $V v^{(1)} = v^{(2)}$, in other words, $\al^{(2)} = \gamma \circ \al^{(1)}$ where $\gamma = \Ad{V}$ is a FDQC with symmetric gate $V$. That is, $\al^{(1)}$ and $\al^{(2)}$ are equivalent as 0-dimensional symmetric entanglers. Since we used stacking to reduce to the case where both entanglers act on the same Hilbert space with the same $G$-action, we conclude that any $\al^{(1)}$ and $\al^{(2)}$ with $\ind(\al^{(1)}) = \ind( \al^{(2)})$ are stably equivalent. This proves the Lemma.
\end{proof}

\begin{lemma} \label{lem.equivalence-block-partitioned-1d}
    Any two block partitioned symmetric entanglers in 1d are stably equivalent. 
    {The auxiliary on-site dimension is bounded in terms of $G$ and $\sup_{j\in\ZZ}\dim(\mathcal A_j)$.} 
\end{lemma}
    
\begin{proof}
    We prove that a block partitioned symmetric entangler $(\alpha, \sfU, \AA)$ is stably equivalent to $(\id, \sfU, \AA)$. As at the beginning of this section, assume without loss of generality that $\al$ is on-site and that its on-site restrictions are given by $\Ad{v_i}$ for unitaries $v_i \in \caA_i$ with $\beta_g(v_i) = \lambda_i(g) v_i$ for some 1-cocycles $\lambda_i$.
    
    We argue that the on-site charges $[\lambda_i] \in H^1(G, U(1))$ can be trivialized by invoking the Eilenberg-Mazur swindle argument, similarly as in the proof of Lemma \ref{lem.triviality-striped-symmetric-entanglers}. Comparing the proof of Lemma \ref{lem.triviality-striped-symmetric-entanglers} to the present context, we have that:
    \begin{enumerate}
        \item A strip $s_{I_j}$ is now a site $j \in \ZZ$. 
        
        \item The 1-cocycles $\lambda_j$ are analogous to the $H^2(G,U(1))$ indices attached to strips in the proof of Lemma \ref{lem.triviality-striped-symmetric-entanglers}.  

        \item In the proof of Lemma \ref{lem.triviality-striped-symmetric-entanglers} a crucial input was the complete classification of 1-dimensional symmetric entanglers by $H^2(G, U(1))$ (i.e.\ Lemma \ref{lem.classification-1d-symmetric-entanglers}).
        The analogous facts in the present context are provided by the complete classification of 0-dimensional symmetric entanglers (Lemma \ref{lem.classification of 0d symmetric entanglers}). {{There, the auxiliary on-site dimension is bounded in terms of $G$ and the dimensions of the algebras being compared.}} 
    \end{enumerate}
    By using this dictionary, the proof of this Lemma is now completely analogous to the one of Lemma \ref{lem.triviality-striped-symmetric-entanglers} above.
\end{proof}

\subsection{$H^2(G,U(1))$ charges in 1d}

Let $(\beta_g = \otimes_i \Ad{U_i(g)}, \AA)$ be a {on-site} LPS on a spin chain. The on-site actions $g \to U_i(g)$ are, in general, projective representations of $G$. The next Lemma states that the on-site $H^2(G,U(1))$ indices are not invariant with respect to stable classification of LPSs. 

\begin{lemma} \label{lem.equivalence-decoupled-1d}
Any two {on-site} LPSs in 1d are stably equivalent. The on-site dimensions of the auxiliary spin systems needed to realize this equivalence are uniformly bounded in terms of $G$ and the original on-site algebra dimension bounds.
\end{lemma}
\begin{proof}
    We prove that any {on-site} LPS $(\beta,\AA)$ is {stably equivalent} to {an action induced by an {on-site unitary representation datum}}. Assume that $\beta_g = \otimes_i \Ad{U_i(g)}$ and denote by $[\mu_i] \in H^2(G,U(1))$ the equivalence class of the projective representation $U_i$. The proof now proceed by applying the the Eilenberg swindle argument in the same way as in the proof of Lemmas \ref{lem.triviality-striped-symmetric-entanglers} and \ref{lem.equivalence-block-partitioned-1d}. In this case, 
    \begin{enumerate}
        \item A strip $s_{I_j}$ is a site $j \in \ZZ$;
        \item The $H^2(G,U(1))$ indices $[\mu_i]$ of the corresponding actions $U_i$ are analogous to the $H^2(G,U(1)$ indices attached to strips in the proof of Lemma \ref{lem.triviality-striped-symmetric-entanglers}. 
        \item {In the proof of Lemma
\ref{lem.triviality-striped-symmetric-entanglers},
the objects carrying the $H^2(G,U(1))$-charges were
1-dimensional symmetric entanglers. Here they are
0-dimensional LPSs. Their complete classification by
the same group, given by Lemma
\ref{lem.classification-0d-LPS}, supplies the corresponding
input to the swindle argument.}
    \end{enumerate}
Now the proof of Lemma \ref{lem.triviality-striped-symmetric-entanglers} can be repeated with the given adjustments to conclude that $(\beta,\AA)$ is stably equivalent to {an action induced by an {on-site unitary representation datum}}.
\end{proof}

\subsection{Equivalence of product $G$-states}

Let $\caA^{\reg}$ be the algebra with on-site algebras $\caA^{\reg}_i = \End( \C[G] )$ and let $\sfU^{(\reg)}=(L_i)_i$ be the symmetry datum where each $L_i$ is the left regular representation as described in Lemma \ref{lem.isomorphism-projective-reps}.  Further, let $| + \rangle  = \abs{G}^{-1/2} \sum_g | g \rangle \in \C[G]$ and 
$$
\phi^{(+)} = \bigotimes_i \left(\langle +|\, \cdot \, | + \rangle\right)
$$
Note that  $(\phi^{(+)},\sfU^{\reg}, \caA^{\reg})$ is a special product $G$-state.

\begin{lemma} \label{lem:alternative equivalence of G-product states}
Consider two product $G$-states $(\phi^{(1)},\sfU,\AA)$ and $(\phi^{(2)},\sfU,\AA)$ on the same $d$-dimensional spin system. 
Then the product $G$-states $\phi^{(1)} \otimes \phi^{(+)}$ and $\phi^{(2)} \otimes \phi^{(+)}$ are related by an {on-site} symmetric entangler.
\end{lemma}

\begin{proof}
It suffices to work at a single site, whose index we suppress.
Let $v,w$ be unit vectors (in a common Hilbert space) representing the (on-site restrictions of) two states $\phi^{(1)},{\phi^{(2)}}$.
Their invariance implies
\[
    U(g)v=\mu(g)v,\qquad U(g)w=\nu(g)w
\]
for characters $\mu,\nu$ of $G$. Put
$\chi=\mu\overline{\nu}$ and $R=U\otimes L$, where $L$ is the left regular representation on the ancilla. 

By Lemma \ref{lem.isomorphism-projective-reps}, applied to
$U$ and $\chi U$, the representations $R$ and $\chi R$
are unitarily equivalent. Thus there is a unitary $V$ such that
\[
    R(g)V=\chi(g)VR(g).
\]
In particular, $V(w\otimes +)$ and $v\otimes +$ both belong
to the subspace of vectors transforming under $R$ according to the character $\mu$.  By composing $V$
on the left with a unitary commuting with $R$, we may
therefore arrange that
\[
    V(w\otimes +)=v\otimes +.
\]
Consequently, $\Ad{V}$ commutes with $\Ad{R(g)}$ and
\[
    (\phi^{(1)}\otimes\phi^{(+)})\circ\Ad{V}
    =\phi^{(2)}\otimes\phi^{(+)}.
\]
Carrying through this construction at every site gives the required
{on-site} symmetric entangler.

\end{proof}

Note that the {on-site} symmetric entangler constructed above is not necessarily an FDQC with symmetric gates. However, by stacking we can make it into one.

\begin{lemma}\label{lem.onsite-stabilization}
Let $d\geq1$. Every on-site symmetric entangler $(\gamma,\sfU,\AA)$ is stably equivalent to the identity. Moreover, there is an auxiliary symmetry datum $(\sfV,\caB)$ admitting a special product $G$-state such that $\gamma\otimes\id_{\caB}$ admits symmetric gates.
\end{lemma}
\begin{proof}
For $d=1$, stable equivalence to the identity follows from Lemma~\ref{lem.equivalence-block-partitioned-1d}. For $d>1$, apply its swindle construction independently along parallel lattice lines. The construction has uniform depth, gate range, and auxiliary on-site dimension, so these circuits can be performed in parallel.

By the definition of stable equivalence, there is an auxiliary symmetry datum $(\sfV,\caB)$ such that $\gamma\otimes\id_{\caB}$ admits symmetric gates. Adjoining a further regular representation at every auxiliary site ensures the existence of a special product $G$-state, since any genuine unitary representation tensored with the regular representation contains the trivial representation.
\end{proof}

The following lemma (also a corollary of \cite[Theorem 4.2]{de_o_carvalho_classification_2025}) is particularly significant because it shows us that the classification problem in $d\geq 1$, contrary to that in $d=0$, concerns spatial entanglement. 
\begin{lemma} \label{lem.equivalence-product-g-states}
    {Let $d\geq1$.}
    Any two product $G$-states $(\phi, \sfU,\AA)$ and $(\phi',\sfU',\AA')$ on $d$-dimensional spin systems $\AA,\AA'$ are stably equivalent.
\end{lemma}
\begin{proof}
First suppose the two product $G$-states are on the same spin system. Lemma~\ref{lem:alternative equivalence of G-product states} gives an on-site symmetric entangler $\eta$ such that
\[
\phi'\otimes\phi^{(+)}=(\phi\otimes\phi^{(+)})\circ\eta.
\]
By Lemma~\ref{lem.onsite-stabilization}, choose an auxiliary system with a special product $G$-state $\sigma$ such that $\eta\otimes\id$ admits symmetric gates. Setting $\widetilde\phi=\phi^{(+)}\otimes\sigma$, we obtain
\[
\phi'\otimes\widetilde\phi
=(\phi\otimes\widetilde\phi)\circ(\eta\otimes\id).
\]
This proves stable equivalence in the case of a common spin system.

For different spin systems, choose special product $G$-states on
$(\sfU'\otimes\sfU^{\reg},\AA'\otimes\caA^{\reg})$ and
$(\sfU\otimes\sfU^{\reg},\AA\otimes\caA^{\reg})$.
Stack the first onto $\phi$ and the second onto $\phi'$. After the equivariant on-site interchange of the $\AA$ and $\AA'$ factors, these are product $G$-states on the same system $\AA\otimes\AA'\otimes\caA^{\reg}$. The case just proved therefore establishes their stable equivalence.
\end{proof}


\paragraph{Acknowledgments.} W.D.R.~and B.O.C.~were supported by the FWO and F.R.S.-FNRS under the Excellence of Science (EOS) programme through the research project G0H1122N EOS 40007526 CHEQS, the KULeuven Runner-up Grant No.~iBOF DOA/20/011, and the internal KULeuven Grant No.~C14/21/086.

\paragraph{Data Availability.} Data sharing is not applicable to this article as no new data were created or analysed in this study.

\paragraph{Conflict of interest.} The authors declare no conflict of interest.

\paragraph{Usage of LLMs.} ChatGPT Sol and ChatGPT Astra were used while preparing this manuscript for resubmission after the first round of reviews. In particular, we used LLMs to vet the manuscript, revealing some gaps in the proofs which we fixed in this new version. Furthermore, the LLMs were used to fix small errors in some of the proofs, and for editorial purposes.

\bibliographystyle{unsrturl}

\bibliography{refs}

\end{document}